\documentclass[journal]{IEEEtran}

\usepackage{amsmath,amssymb,amsfonts}
\usepackage{mathtools}
\usepackage{bm}
\usepackage[hidelinks]{hyperref}
\usepackage{subcaption}
\usepackage{stfloats}
\usepackage{xfrac}
\usepackage{comment}
\usepackage{graphicx}
\usepackage{pgfplots}
\usepackage{cite}
\usepackage{relsize}
\usepackage{booktabs}
\pgfplotsset{compat=1.15}

\allowdisplaybreaks

\newcommand{\E}{\mathbb{E}}

\newcommand{\argmin}{\mathop{\mathrm{argmin}}}
\newcommand{\rxx}{r(x)}

\newcommand{\betaN}{\beta}
\newcommand{\Fone}[3]{{}_1F_1\left(#1;#2;#3\right)}
\newcommand{\Wk}{w_k}

\newtheorem{lemma}{Lemma}

\newtheorem{proposition}{Proposition}
\newtheorem{remark}{Remark}

\begin{document}

\title{Constellation Design for Robust Interference Mitigation}

\author{Athanasios T. Papadopoulos, Thrassos K. Oikonomou,~\IEEEmembership{Graduate Student Member,~IEEE},\\Dimitrios Tyrovolas,~\IEEEmembership{Member,~IEEE}, Sotiris A. Tegos,~\IEEEmembership{Senior Member,~IEEE}, \\Panagiotis D. Diamantoulakis,~\IEEEmembership{Senior Member,~IEEE}, Panagiotis Sarigiannidis,~\IEEEmembership{Member,~IEEE}, \\and George K. Karagiannidis,~\IEEEmembership{Fellow,~IEEE}
\thanks{A. T. Papadopoulos, T. K. Oikonomou, D. Tyrovolas, S. A. Tegos, P. D. Diamantoulakis, and G. K. Karagiannidis are with the Department of Electrical and Computer Engineering, Aristotle University of Thessaloniki, 54124 Thessaloniki, Greece (e-mails: atpapadop@ece.auth.gr, toikonom@ece.auth.gr, tyrovolas@auth.gr, tegosoti@auth.gr, padiaman@auth.gr, geokarag@auth.gr).}
\thanks{P. Sarigiannidis is with the Department of Electrical and Computer Engineering, University of Western Macedonia, 50100 Kozani, Greece (e-mail: psarigiannidis@uowm.gr).}
}
\maketitle

\begin{abstract}
This paper investigates symbol detection for single-carrier communication systems operating in the presence of additive interference with Nakagami-$m$ statistics. Such interference departs from the assumptions underlying conventional detection methods based on Gaussian noise models and leads to detection mismatch that fundamentally affects symbol-level performance. In particular, the presence of random interference amplitude and non-uniform phase alters the structure of the optimal decision regions and renders standard Euclidean distance-based detectors suboptimal. To address this challenge, we develop the maximum-likelihood Gaussian-phase approximate (ML-G) detector, a low-complexity detection rule that accurately approximates maximum-likelihood detection while remaining suitable for practical implementation. The proposed detector explicitly incorporates the statistical properties of the interference and induces decision regions that differ significantly from those arising under conventional metrics. Building on the ML-G framework, we further investigate constellation design under interference-aware detection and formulate an optimization problem that seeks symbol placements that minimize the symbol error probability subject to an average energy constraint. The resulting constellations are obtained numerically and adapt to the interference environment, exhibiting non-standard and asymmetric structures as interference strength increases. Simulation results demonstrate clear symbol error probability gains over established benchmark schemes across a range of interference conditions, particularly in scenarios with dominant additive interference.
\end{abstract}

\section{Introduction}\label{sec:intro}
Modern wireless communication systems increasingly operate in environments where multiple heterogeneous radio systems occupy overlapping spectral resources \cite{spectrum_sharing_v1, spectrum_sharing_v2, spectrum_sharing_v3, spectrum_sharing_v4, spectrum_sharing_6g}. In such settings, however, receivers are often exposed to interference generated by coexisting emitters. A prominent example arises from the interaction between communication receivers and sensing systems, which share frequency bands in both civilian and defense applications \cite{spectrum_sharing_civilian, interference_defense, interference_defense_2}. Specifically, sensing signals are typically characterized by high peak power, structured waveforms, and strong temporal and spectral correlations, all of which fundamentally differentiate them from thermal noise \cite{cook2012radar, joint_radar}. As a result, interference cannot be accurately modeled using conventional additive white Gaussian noise (AWGN) models, particularly in scenarios where the interference power is comparable to or exceeds that of the communication signal. Therefore, understanding the impact of such interference and incorporating it into receiver designs that account for its effects has emerged as a key challenge for reliable communication in modern wireless systems.

In the presence of strong and structured interference, reliable communication depends on how accurately its impact is reflected at the receiver. In this context, the interference observed after demodulation is shaped not only by the characteristics of the interfering transmitter, but also by the propagation environment and the structure of the emitted waveform \cite{Nartasilpa2018}. Specifically, multipath scattering and shadowing give rise to random fluctuations in the interference amplitude, while at the same time deterministic waveform features, together with their interaction with the channel, lead to structured and non-uniform phase behavior \cite{radar_multipath, statistical_modeling_radar}. As a consequence, the resulting interference cannot be regarded as an isotropic disturbance in the complex plane, but instead introduces direction-dependent distortions that affect different signal components unequally \cite{coexistence, performance}. These distortions, in turn, directly influence symbol separability and can lead to significant performance degradation when they are not explicitly accounted for at the receiver. Therefore, this highlights the need to develop detection schemes that incorporate interference statistics shaped by realistic propagation phenomena, rather than relying on conventional detectors that remain agnostic to the underlying interference structure \cite{survey_v1}.

\subsection{State-of-the-Art}
Many works in the existing literature have investigated the impact of strong interference on communication system performance, with many efforts focusing on adapting receiver operation by exploiting the available degrees of freedom to mitigate its effects. In this direction, the authors of \cite{roy, hangzheng2015communication} exploited the temporal characteristics of interference by detecting deterministic interference bursts or deliberately avoiding time intervals dominated by external emissions. However, while such strategies can be effective in specific operating regimes, their applicability depends on reliable temporal information and often incurs spectral efficiency losses due to time resource allocation for interference monitoring, as well as increased coordination requirements between coexisting systems. Building on the idea of receiver adaptation, robust MIMO-based coexistence strategies have also been proposed, in which transmit or receive designs are optimized against bounded, sparse, or worst-case interference models \cite{li_V2_2017joint}. Although such approaches can provide interference suppression through spatial processing, they rely on strict antenna coordination and sustained observation of the interference environment, which can be difficult to maintain in dynamic propagation conditions. Taken together, temporal and spatial processing primarily rely on external structure or auxiliary resources and do not fully characterize the role of interference at the point where symbol decisions are formed.

In contrast to approaches that mitigate interference by exploiting external degrees of freedom, a complementary line of work has emphasized the role of detection metrics that explicitly account for stochastic distortions at the symbol level. Specifically, classical Euclidean distance detection has been revisited in the presence of phase noise, with several studies demonstrating that decision rules matched to the statistical properties of the phase noise can significantly enhance resilience and improve symbol error performance. In particular, the authors of \cite{foschini} and \cite{soft_metrics} derived maximum likelihood detection metrics that incorporate the underlying phase noise distribution, resulting in decision rules that depart from conventional Euclidean distance criteria and reduce detection mismatch under random phase perturbations. Beyond pure phase noise, more general impairment models have also been considered, where the transmitted symbol undergoes stochastic multiplicative distortions jointly with random phase perturbations, leading to modified detection rules and performance gains compared to mismatched receivers \cite{thrassos}. In this broader framework of symbol-aware detection, interference mitigation has also been examined in the context of radar coexistence, where the authors of \cite{Nartasilpa2018} modeled the interfering signal as a deterministic component with uniformly distributed random phase, yielding tractable likelihood expressions and simplified detectors. Building on this detection-centric perspective, subsequent works have shown that such interference-aware detection metrics can also serve as a foundation for constellation design, where the signal alphabet is optimized with respect to the induced likelihood structure rather than conventional distance-based criteria \cite{nartasilpa_const}. More recently, data-driven approaches, including deep learning based schemes, have been proposed to learn symbol decision rules directly from observations corrupted by radar-induced interference \cite{liu2022deep, deep2}. Collectively, these works highlight that the formulation of the detection rule itself constitutes a powerful design dimension through which stochastic distortions and interference can be internalized, directly shaping the resulting decision regions, achievable performance, and the criteria according to which signal constellations can be optimized.


\subsection{Motivation \& Contribution}
In communication systems operating in the presence of interference, robustness fundamentally depends on the ability of the receiver to capture the statistical characteristics of the interference and embed them into the detection process. However, many existing detection methods rely on simplified interference models that do not adequately reflect the conditions encountered in modern spectrum sharing environments, particularly when interference is strong and exhibits pronounced structure. In particular, interference is either absorbed into an effective AWGN representation or modeled using a deterministic amplitude together with a uniformly distributed phase, thereby neglecting the joint stochastic behavior of interference amplitude and phase induced by realistic propagation effects \cite{Nartasilpa2018}. At the same time, data-driven detection techniques rely on extensive training across operating conditions, which can become a practical bottleneck in dynamic environments \cite{liu2022deep, deep2}. To the best of the authors’ knowledge, a receiver-side detection framework that explicitly accounts for jointly random interference amplitude and phase and that can be directly leveraged for interference-aware constellation design has not been systematically developed.

In this paper, we develop an approximate maximum likelihood detection framework for communication systems operating in the presence of stochastic interference with random amplitude and non-uniform phase statistics. The proposed framework systematically incorporates known statistical models of the interference into the receiver design, departing from conventional interference modeling assumptions. Specifically, the contributions of this work are summarized as follows:
\begin{itemize}
    \item We derive an approximate maximum likelihood detector that yields a tractable decision metric capturing the joint statistical impact of interference amplitude and phase, and admits a series representation with well-defined convergence properties that enable computationally efficient implementation.
    \item We analyze the geometric structure of the corresponding decision regions and characterize the resulting decision boundaries in the complex plane, providing insights into how interference in amplitude and phase jointly influences the decision geometry.
    \item We formulate and solve a signal constellation optimization problem aimed at minimizing the average symbol error probability (SEP) under interference conditions. The optimization is carried out using a global search strategy combined with local refinement under average energy constraints, leading to constellation geometries that adapt to the interference environment.
\end{itemize}

\subsection{Structure}
The remainder of this paper is organized as follows. Section II presents the system model and the statistical interference framework. Section III introduces the proposed detection metric and its tractable formulation. Section IV considers the resulting decision geometry and the associated constellation optimization problem. Section V provides numerical results and performance comparisons. Section VI concludes the paper.

\section{System Model}\label{sec:sys_model}
We consider a single-carrier communication link affected by a narrowband interferer. The complex baseband observation over one symbol interval is expressed as
\begin{equation}\label{eq:system_model}
\begin{aligned}
    Y = \sqrt{S}X + I + N,
\end{aligned}
\end{equation}
where $X$ denotes the transmitted symbol drawn from a constellation $\mathcal{C}$ with unit average energy, i.e., $\mathbb{E}[|X|^2] = 1$, $S$ represents transmit power, and $N$ denotes additive noise, modeled as a circularly symmetric complex Gaussian random variable with zero mean and unit variance, i.e., $N\sim\mathcal{CN}(0,1)$. Moreover, the interference term $I = A e^{j\Theta}$ is modeled as a complex random variable whose amplitude and phase jointly follow a Nakagami-$m$ envelope-phase distribution \cite{nakagami-phase}. In particular, the envelope $A$ follows a Nakagami-$m$ distribution with shape parameter $m$ and spread parameter $\Omega$, with probability density function (PDF)
\begin{equation}\label{eq:nakamp}
f_A(a)=\frac{2m^m}{\Gamma(m)\Omega^{m}} a^{2m-1}
e^{-\frac{m}{\Omega}a^{2}},\quad a\ge0,
\end{equation}
where $\Gamma(\cdot)$ denotes the Gamma function \cite{gradshteyn2007}, 
while the phase $\Theta$ follows the corresponding Nakagami-$m$ phase distribution, whose PDF is given by
\begin{equation}\label{eq:nakPhase}
f_{\Theta}(\theta)
= \frac{\bigl|\sin 2\theta\bigr|^{m-1}}
{2\sqrt{\pi}C(m)},\quad \theta\in[0,2\pi],
\end{equation}
where the normalizing constant is $C(m)=\Gamma\left(\tfrac{m}{2}\right)\big/\Gamma\left(\tfrac{m+1}{2}\right)$. Therefore, using the same shape parameter $m$ for both interference amplitude and phase ensures that the two components reflect a common scattering environment in which fading severity and angular dispersion vary consistently. Based on this model, the signal-to-noise ratio (SNR), denoted by $\gamma_S$, is defined as the average received signal power relative to the noise power, while the interference-to-noise ratio (INR), denoted by $\gamma_I$, is defined as the average interference power relative to the noise power.
Finally, we define the signal-to-interference ratio as the ratio between SNR and INR, and denote it by $\gamma$.

\section{Detection Framework}\label{sec:detector-design}
In this section, we develop an analytical framework for symbol detection in the presence of Nakagami-$m$ fading interference. Specifically, an approximate maximum likelihood decision rule is constructed by incorporating the statistical properties of the interference and approximating its phase distribution through Gaussian moment matching, resulting in a closed-form and computationally tractable detection metric. Finally, the resulting metric defines the decision regions and shows how the interference statistics deform the boundaries between constellation symbols.
\subsection{Metric Derivation}
In this subsection, we focus on the derivation of a detection metric that accounts for the statistical structure of the interference affecting the received signal. In more detail, following the maximum likelihood principle, symbol decisions are formed by evaluating the likelihood of the observation $Y$ under the joint statistics of the interference amplitude $A$ and phase $\Theta$. Accordingly, the conditional likelihood function $f(y|x)$ can be expressed as
\begin{equation}\label{eq:conditional_prob_ML}
    \begin{aligned}
        f(y|x) = \int_{-\pi}^{\pi}\int_{0}^{\infty}f_{Y|X,A,\Theta}\left(y|x,a,\theta\right)f_{A}(a)f_{\Theta}(\theta)dad\theta,
    \end{aligned}
\end{equation}
where $f_{Y|X,A,\Theta}\left(y|x,a,\theta\right)$ is given as \cite{Nartasilpa2018}
\begin{equation}\label{eq:condPDF}
\begin{aligned}
      f_{Y|X,A,\Theta}\left(y|x,a,\theta\right)=\frac{1}{\pi}\exp\left(-|y-\sqrt{S}x-a e^{j\theta}|^{2}\right).
\end{aligned}
\end{equation}
As can be observed, evaluating \eqref{eq:conditional_prob_ML} requires averaging the conditional likelihood over the distributions of $A$ and $\Theta$. While the averaging with respect to the interference amplitude yields a closed-form expression, the phase distribution plays a dominant role in shaping the likelihood and calls for an analytically convenient representation that preserves its essential dispersion characteristics. Accordingly, the phase distribution is approximated by a Gaussian distribution through moment matching \cite{TYROVOLASHARQ}
\begin{lemma}\label{lemma:MoM}
The phase random variable $\Theta$ in \eqref{eq:nakPhase} can be well-approximated by a Gaussian distribution with mean $\pi$ and variance
\begin{equation}\label{eq:sigTheta-exact}
\sigma_{\Theta}^2(m)
=\frac{\pi^{2}}{4}
+
\frac{\displaystyle\int_{0}^{\pi} t^{2}\sin^{m-1}tdt}
     {4\sqrt{\pi}C(m)},
\end{equation}
where 
\end{lemma}
\begin{IEEEproof}
    The proof is given in Appendix~\ref{sec:appendix_MoM}.
\end{IEEEproof}

Using the Gaussian moment-matching approximation of the phase distribution introduced above, the likelihood function in \eqref{eq:conditional_prob_ML} admits a tractable representation that leads to an approximate maximum likelihood detection rule, presented in the following proposition
\begin{proposition}\label{prop:ML_detector}
The transmitted symbol $x \in \mathcal{C}$, when affected by AWGN and by interference following a Nakagami-$m$ envelope with Nakagami-$m$ phase statistics, can be detected using the maximum-likelihood Gaussian-phase approximate (ML-G) rule, given by
    \begin{equation}\label{eq:ML_final}
    \begin{aligned}
        \hat x_{\text{ML-G}}(y)
        = \argmin_{x\in\mathcal X}\hspace{0.25em} r^{2}
        - \ln \mathcal{S}\bigl(r,\phi\bigr),
    \end{aligned}
\end{equation}
where $r = |y-\sqrt{S}x|$, $\phi = \arg\{y-\sqrt{S}x\}$ and
    \begin{equation}\label{eq:Sdef}
        \begin{aligned}
            \mathcal{S}(r,\phi)=I_{0,m}(r) +\sum_{k=1}^{\infty} \Wk(\phi) I_{k,m}(r), 
        \end{aligned}
    \end{equation}
with 
    \begin{equation}
        \begin{aligned}
            & I_{k,m}(r)= \frac{\Gamma\left(m+\tfrac{k}{2}\right)}{2\betaN^{m+\frac{k}{2}}k!}r^{k}
\Fone{m+\frac{k}{2}}{k+1}{\frac{r^{2}}{\betaN}},
        \end{aligned}
    \end{equation}
    $\Wk(\phi)= (-1)^{k}2e^{-\frac{1}{2}\sigma_\Theta^{2}k^{2}}\cos(k\phi)$, $\betaN= 1+\tfrac{m}{\Omega}$, and $\Fone{\cdot}{\cdot}{\cdot}$ denotes the confluent hypergeometric function of the first kind \cite{gradshteyn2007}.
\end{proposition}
\begin{IEEEproof}
    The proof is provided in Appendix \ref{sec:appendix_ML_detector}.
\end{IEEEproof}

\begin{remark}
\label{rem:impl-limits}
Since the weights $w_k(\phi)$ decay exponentially with $k$, the series defining $\mathcal{S}(r,\phi)$ in \eqref{eq:Sdef} is absolutely convergent for all $r$ and $\phi$, and in practice can be accurately approximated using only a small number of terms, ensuring numerical stability and analytical tractability of the detector.
\end{remark}

Building on the convergence properties of the series representation of 
$\mathcal{S}(r,\phi)$ in~\eqref{eq:Sdef}, we now examine its behavior over the range of values relevant to the detection metric. In particular, the evaluation of $\mathcal{S}(r,\phi)$ involves the residual magnitude $r = |y-\sqrt{S}x|$, which is confined to a bounded region in practice. To capture this region, let $R_{\max}>0$ denote a design radius corresponding to the largest value of $r$ under consideration, and let $\varepsilon \in (0,1)$ be a prescribed tolerance that specifies an acceptable approximation accuracy. In this region of interest $r \in [0,R_{\max}]$, we consider a finite truncation of the series in~\eqref{eq:Sdef} such that the aggregate contribution of the remaining terms does not exceed~$\varepsilon$. In this direction, the following proposition establishes that such a truncation index always exists.
\begin{proposition}\label{proposition:existence-trunc}
For any given shape parameter $m>0$, $\Omega>0$, and tolerance 
$\varepsilon\in(0,1)$, there exists a finite truncation index $K$ such that
\begin{equation}
    \sum_{k>K}\left|w_k(\phi)\right|I_{k,m}(r)\le\varepsilon,
    \quad \forall r\in[0,R_{\max}],\ \phi\in\mathbb{R}.
\end{equation}
\end{proposition}
\begin{IEEEproof}
The proof is provided in Appendix \ref{sec:appendix_existence}.
\end{IEEEproof}
Having established that a finite truncation index guaranteeing any prescribed accuracy always exists, we next derive a closed-form expression for obtaining the smallest truncation index $K$.
\begin{proposition}\label{proposition:min-K}
Among all truncation indices $K$ for which the tail of the series in 
\eqref{eq:Sdef} satisfies 
$\sum_{k>K} |w_k(\phi)| I_{k,m}(r) \le \varepsilon$ over the region 
$r\in[0,R_{\max}]$, the smallest guaranteed $K$ is determined implicitly by
\begin{equation}\label{eq:min_K}
\frac{W(R_{\max},\beta,\sigma_{\Theta}, K)\Gamma\left(m+\frac{K+1}{2}\right)}%
{\left(1-q_{K}(R_{\max})\right)\beta^{m+\tfrac{K+1}{2}}(K+1)!}
= \varepsilon,
\end{equation}
where
\begin{equation} \label{defW}
    W(r,\beta, \sigma_{\Theta},K)
    = e^{-\frac{1}{2}\sigma_{\Theta}^{2}(K+1)^{2}}
      r^{K+1}
      \exp\left(\frac{r^{2}}{\beta}\right),
\end{equation}
and
\begin{equation}\label{eq:qK-main}
    q_{K}(r)
    = \frac{re^{-\sigma_{\Theta}^{2}\left(K+\tfrac{1}{2}\right)}}%
           {(K+1)\sqrt{\beta}}
      \sqrt{m+\tfrac{K}{2}+\tfrac{1}{2}}
      \exp\left(\frac{r^{2}}{\beta}\right),
\end{equation}
with $\beta = 1+\tfrac{m}{\Omega}$.
\end{proposition}

\begin{IEEEproof}
The proof is provided in Appendix \ref{sec:appendix_minK}.
\end{IEEEproof}

\begin{figure*}[t]
\centering
\setlength{\tabcolsep}{4pt}
\renewcommand{\arraystretch}{1.1}

\begin{subfigure}[t]{0.33\textwidth}
  \centering
  \includegraphics[width=\linewidth]{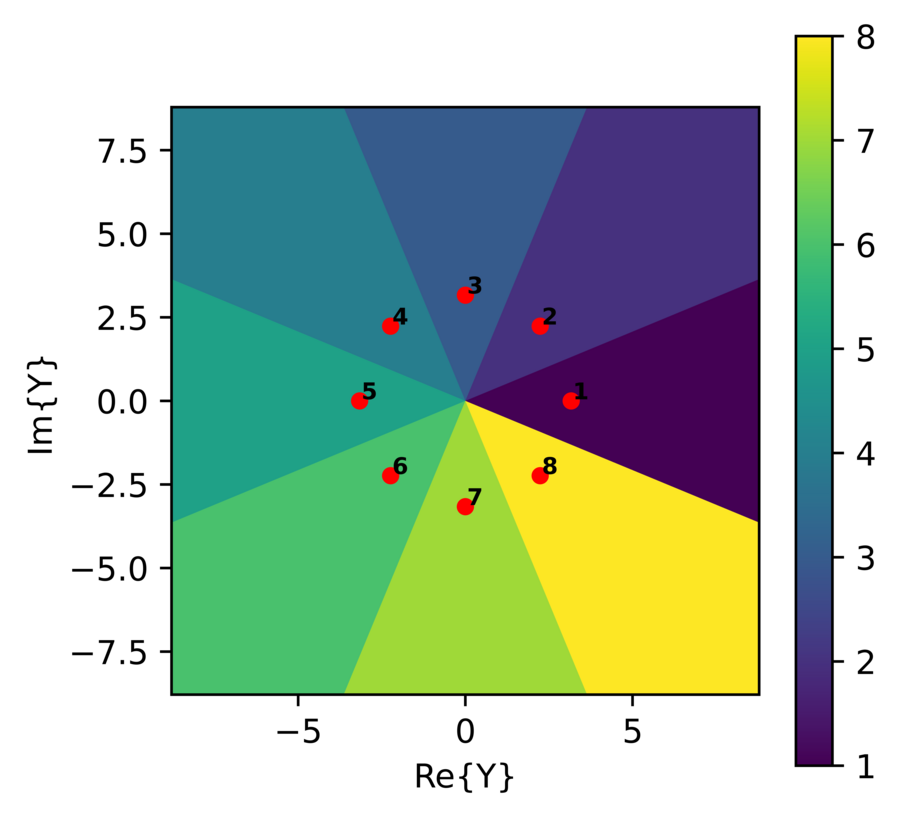}
  \subcaption{CAI, $\gamma_I=2.5$ dB}
\end{subfigure}\hfill
\begin{subfigure}[t]{0.33\textwidth}
  \centering
  \includegraphics[width=\linewidth]{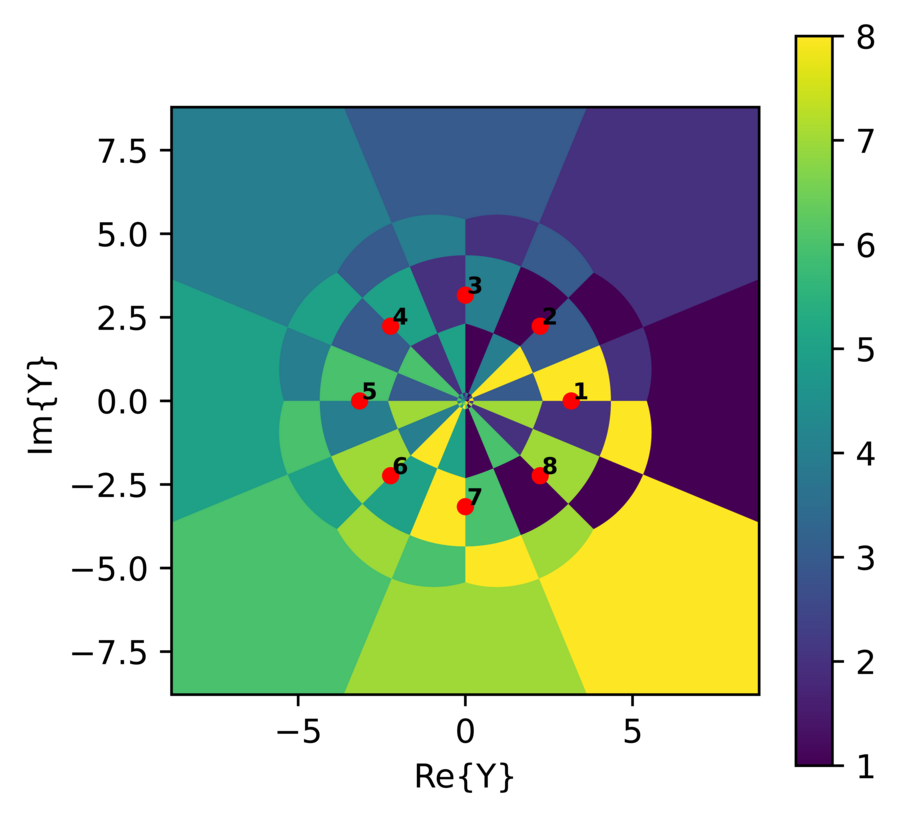}
  \subcaption{CAI, $\gamma_I=10$ dB}
\end{subfigure}\hfill
\begin{subfigure}[t]{0.33\textwidth}
  \centering
  \includegraphics[width=\linewidth]{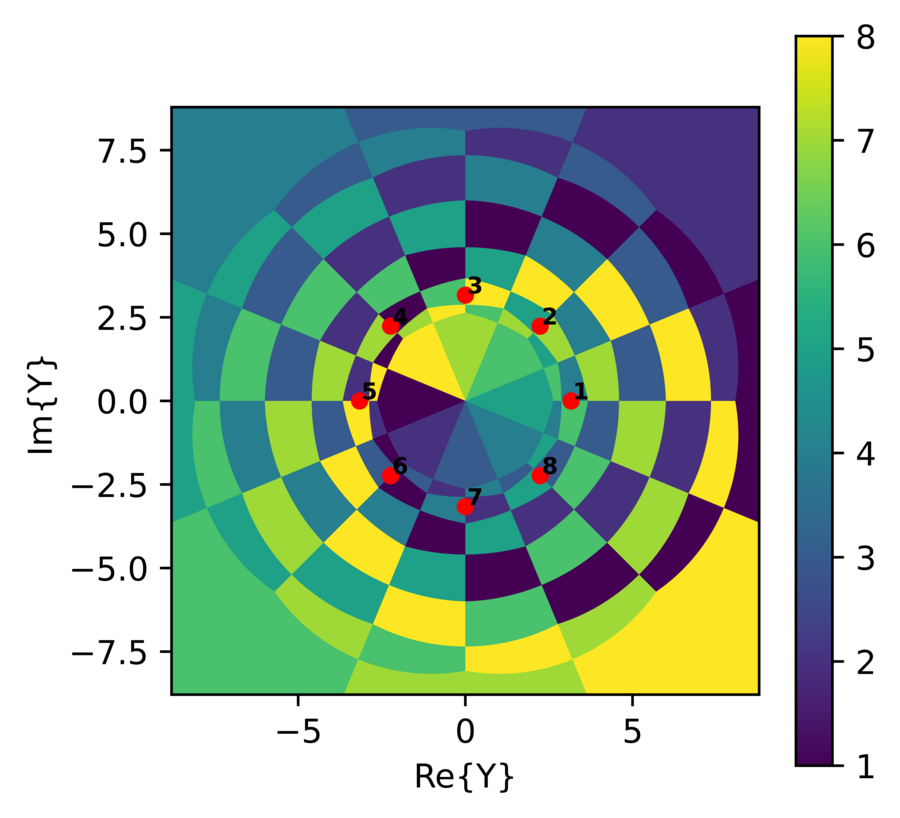}
  \subcaption{CAI, $\gamma_I=15$ dB}
\end{subfigure}

\vspace{4pt}

\begin{subfigure}[t]{0.33\textwidth}
  \centering
  \includegraphics[width=\linewidth]{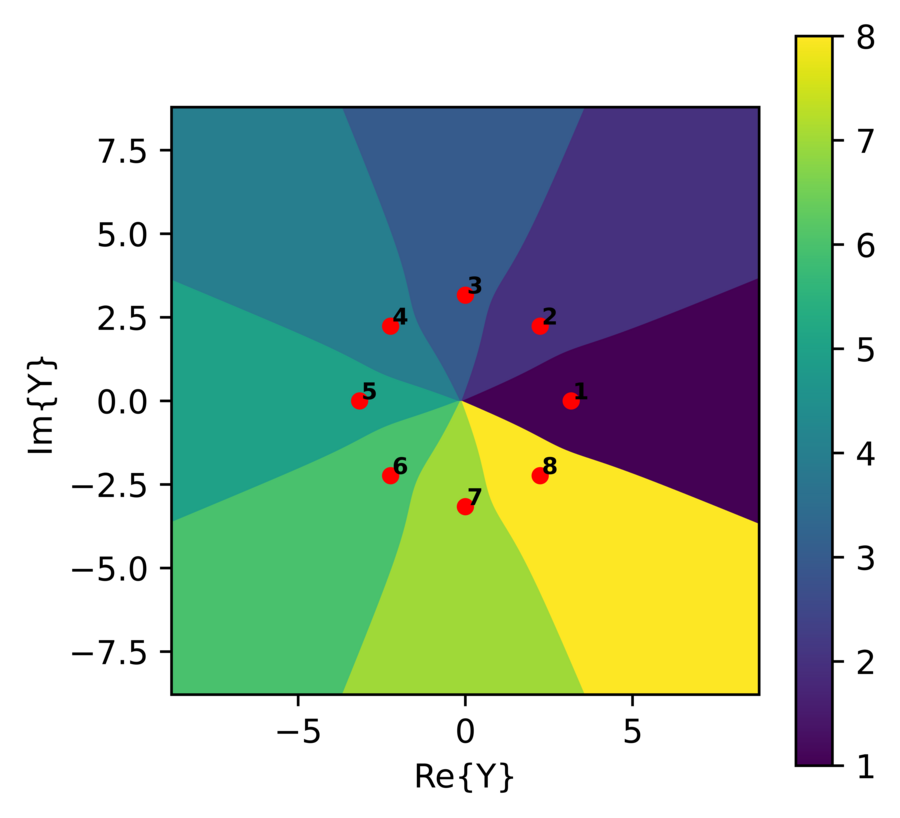}
  \subcaption{ML-G, $\gamma_I=2.5$ dB}
\end{subfigure}\hfill
\begin{subfigure}[t]{0.33\textwidth}
  \centering
  \includegraphics[width=\linewidth]{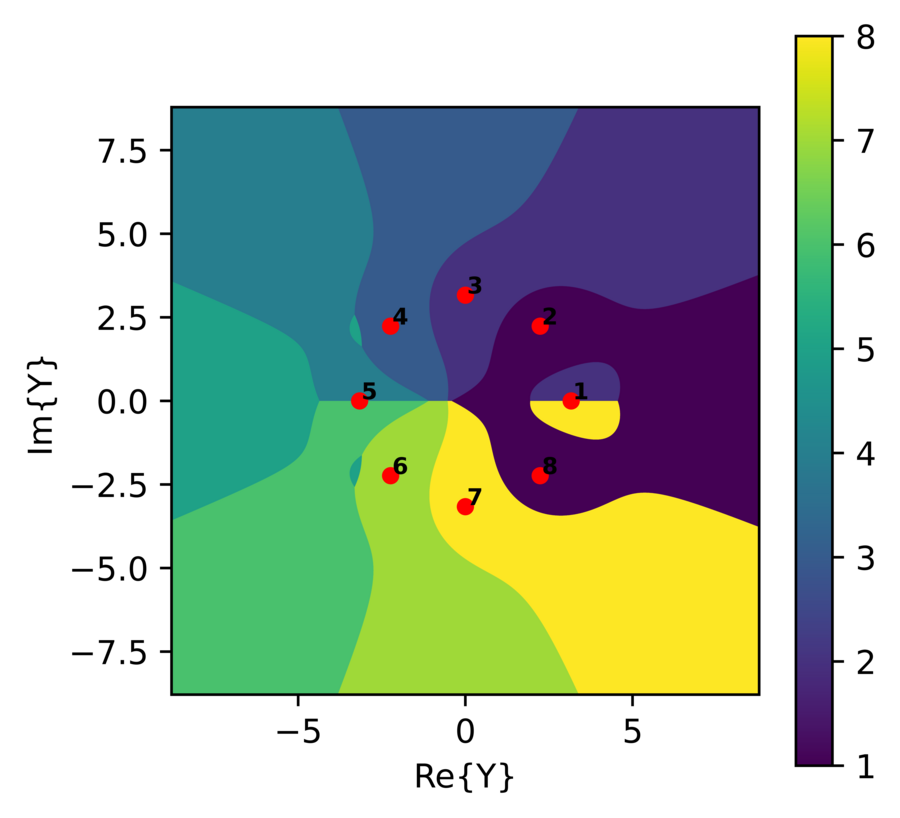}
  \subcaption{ML-G, $\gamma_I=10$ dB}
\end{subfigure}\hfill
\begin{subfigure}[t]{0.33\textwidth}
  \centering
  \includegraphics[width=\linewidth]{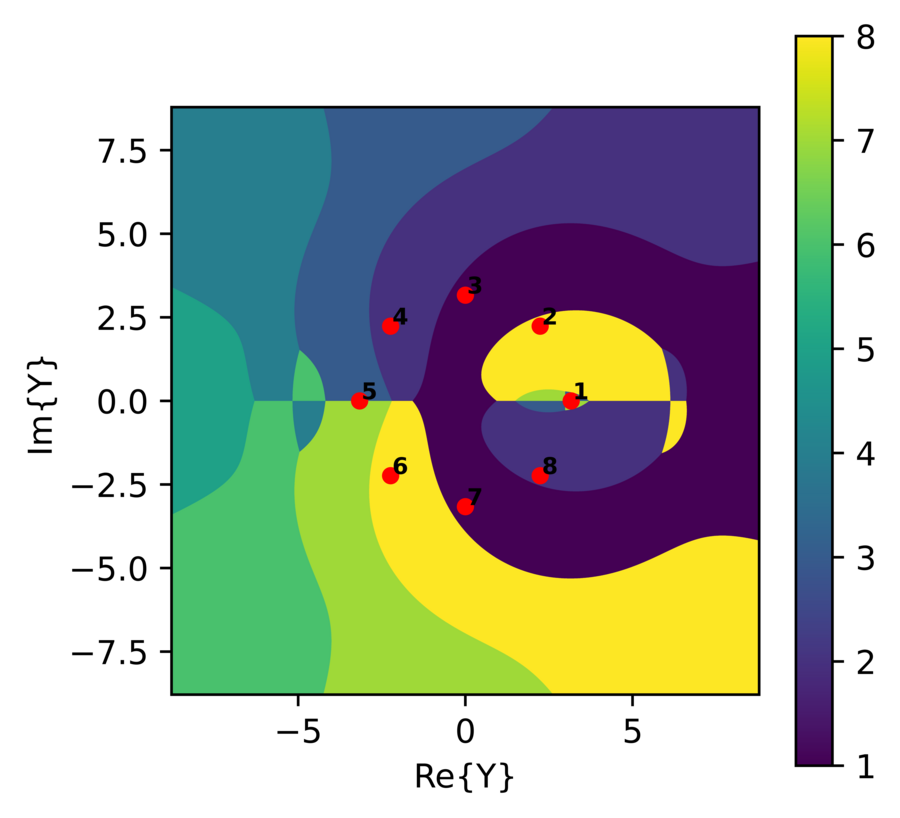}
  \subcaption{ML-G, $\gamma_I=15$ dB}
\end{subfigure}

\caption{Decision regions for 8-PSK at $\gamma_S=10$ dB across various $\gamma_I$.}
\label{fig:DR_3x3_rotated}
\end{figure*}

The above result ensures that the infinite series in~\eqref{eq:Sdef} can be approximated with arbitrary accuracy by retaining a finite number of terms, guaranteeing the numerical tractability of the proposed detector in \eqref{eq:ML_final}. In practice, the truncation index $K$ can be selected to satisfy a prescribed tolerance $\varepsilon$ for a chosen design radius $R_{\max}$. For instance, Table \ref{tab:K_needed_e-3} illustrates the minimum truncation index $K$ required to satisfy a tolerance of $\varepsilon= 10^{-3}$ within a design radius $R_{\max} = 4$. The corresponding values, obtained from~\eqref{eq:min_K}, guarantee that the truncation error remains below $10^{-3}$ uniformly for all $r\in[0,4]$. Table \ref{tab:K_needed_e-3} presents representative combinations of the shape parameter $m$ and $\gamma_I$, showing that only a few terms are sufficient across a broad range of parameters. It should be highlighted that, in many cases, it suffices for one term for the infinite series to converge. This confirms the rapid convergence of the series in~\eqref{eq:Sdef} and highlights the computational efficiency of the proposed detection metric.

\begin{table}
\centering
\caption{Truncation index $K$ for $\varepsilon<10^{-3}$ and $R_{\max}=4$.}
\label{tab:K_needed_e-3}
\setlength{\tabcolsep}{8pt}
\renewcommand{\arraystretch}{1.05}
\begin{tabular}{c|rrrrrr}
\toprule
\multicolumn{7}{c}{Truncation index $K$} \\   
\midrule
$m$ $\backslash$ INR (dB) & $-20$ & $-10$ & $0$ & $10$ & $20$ & $30$ \\
\midrule
2   & 1 & 1 & 2 & 4 & 5 & 5 \\
3   & 1 & 1 & 2 & 4 & 5 & 5 \\
5   & 1 & 1 & 1 & 4 & 5 & 5 \\
10  & 1 & 1 & 1 & 4 & 5 & 5 \\
50  & 1 & 1 & 1 & 6 & 10 & 10 \\
\bottomrule
\end{tabular}
\end{table}

To further interpret the behavior of the detector, it is helpful to consider interference conditions that yield analytical simplifications and offer clearer physical insight. A particularly meaningful case arises when the radar envelope follows a Rayleigh distribution, corresponding to $m=1$, which represents a fully diffuse scattering environment characterized by the absence of a dominant line-of-sight component. Under this condition, the interference statistics are substantially simplified and the proposed detector collapses to the conventional Euclidean distance rule, thereby validating the internal consistency of the formulation and reinforcing its physical interpretability.

\begin{proposition}
\label{prop:Rayleigh}
    When the radar interference follows a Rayleigh envelope, corresponding to $m=1$, 
the proposed detector in~\eqref{eq:ML_final} reduces exactly to the conventional 
Euclidean distance rule given by 
\begin{equation}\label{eq:ML_rayleigh}
    \begin{aligned}
        \hat x^{(m=1)}(y)
        = \argmin_{x\in\mathcal X}\hspace{0.25em} r^{2} .
    \end{aligned}
\end{equation}
\end{proposition}
\begin{IEEEproof}
For $m=1$, the phase law in~\eqref{eq:nakPhase} becomes uniform over 
$[0,2\pi)$, and the interference amplitude reduces to the Rayleigh distribution
$f_A(a) = \frac{2}{\Omega} a e^{-a^{2}/\Omega}$, 
where $\Omega = \mathbb{E}[A^{2}]$ denotes the average interference power. 
Substituting these distributions into the conditional likelihood of $Y$ given $X$ and $A$, 
and then averaging over the Rayleigh envelope amplitude, we obtain
\begin{equation}
    \begin{aligned}
        f_{Y|X}(y|x)
=\int_0^\infty f_{Y|X,A}(y|x,a)f_A(a)da.
    \end{aligned}
\end{equation}
Since for uniform phase $f^{(m=1)}_{Y|X,A}(y|x,a)=\tfrac{1}{\pi}e^{-r^2-a^2}I_0(2ar)$, where where $I_{0}(\cdot)$ denotes the modified Bessel function of the first kind and order zero \cite{gradshteyn2007},
substituting $f_A(a)=\tfrac{2}{\Omega}ae^{-a^2/\Omega}$, we obtain
\begin{equation}
    \begin{aligned}
        f_{Y|X}(y|x)
\propto e^{-r^2}\int_0^\infty ae^{-\beta a^2}I_0(2ar)da.
    \end{aligned}
\end{equation}
Using \cite[6.631/7]{gradshteyn2007},
we derive
\begin{equation}
    \begin{aligned}
        f_{Y|X}(y|x)\propto 
\exp\left(-\tfrac{r^2}{1+\Omega}\right),
    \end{aligned}
\end{equation}
which concludes the proof.
\end{IEEEproof}

\subsection{Decision Regions}\label{sec:regions}
In this section, we characterize the decision regions associated with the proposed ML-G detector. Our analysis reveals how non-uniform interference modifies the decision boundaries between constellation symbols and alter their relative separability.

For a generic detector with metric $\Lambda(y,x)$, the decision region 
associated with symbol $x_i$ is defined as
\begin{equation}\label{eq:region-def}
  \mathcal{D}_i
  = \bigl\{ y\in\mathbb{C} : 
      \Lambda(y,x_i) \le \Lambda(y,x_j),\ \forall j\neq i \bigr\}.
\end{equation}
Equivalently, the pairwise boundary between symbols $x_i$ and $x_j$ is 
determined by the zero-level set of the metric difference
\begin{equation}\label{eq:pairwise-def}
  \Delta_{i,j}(y)
  = \Lambda(y,x_i)-\Lambda(y,x_j) = 0.
\end{equation}
In this direction, using \eqref{eq:ML_final}, we can express the pairwise difference as
\begin{equation}\label{eq:pairwise-G}
\begin{aligned}
  \Delta_{i,j}^{(G)}(y)
  =& \bigl|y-\sqrt{S}x_i\bigr|^{2}
   - \bigl|y-\sqrt{S}x_j\bigr|^{2} \\
  &- \ln\left(\frac{\mathcal{S}(r_j,\phi_j)}{\mathcal{S}(r_i,\phi_i)}\right),
\end{aligned}
\end{equation}
where $(r_\ell,\phi_\ell)$ correspond to the polar coordinates of 
$y-\sqrt{S}x_\ell$, for $\ell\in\{i,j\}$. As can be observed, the first two terms in \eqref{eq:pairwise-G} yield the familiar Voronoi boundaries of conventional Euclidean distance detector, while the additional log-ratio term introduces a symbol-dependent radial correction affected by the interference-induced statistic $\mathcal{S}(r,\phi)$. As a result, the regions $\mathcal{D}_i^{(G)}$ deviate from standard Voronoi 
cells, reflecting asymmetric warping effects that depend jointly on the interference amplitude distribution and the phase dispersion captured by $\mathcal{S}$.

Fig.~\ref{fig:DR_3x3_rotated} provides a visual illustration of the decision regions produced by the proposed ML-G detector. The figure shows the boundaries for an 8-phase shift keying (PSK) constellation at $\gamma_s = 10$ dB under three interference levels $\gamma_I \in \{2.5, 10, 15\}$ dB, with the constellation points indicated by red markers, and includes a comparison with the constant amplitude interference (CAI) detection metric described in \cite{Nartasilpa2018}. As can be seen, when the interference is weak, e.g., $\gamma_I = 2.5$ dB, INR is much smaller than SNR and the resulting regions are nearly identical to those of the conventional Euclidean distance detector, since the decision geometry is dominated by the noise term. However, as the interference increases and becomes comparable to or stronger than the desired signal, as in the cases $\gamma_I = 10$ dB and $\gamma_I = 15$ dB, the regions produced by the Gaussian-phase metric deviate substantially from the CAI boundaries and display clear warping and asymmetry. This behavior reflects the influence of the trigonometric cosine factors in the harmonic weights $w_k$ of \eqref{eq:Sdef}, which modulate the function $\mathcal{S}(r,\phi)$ and thus reshape the decision boundaries.

\section{Constellation Design}

\begin{figure*}[!ht]
\centering

\begin{subfigure}[t]{0.33\textwidth}
  \centering
  \includegraphics[width=\linewidth]{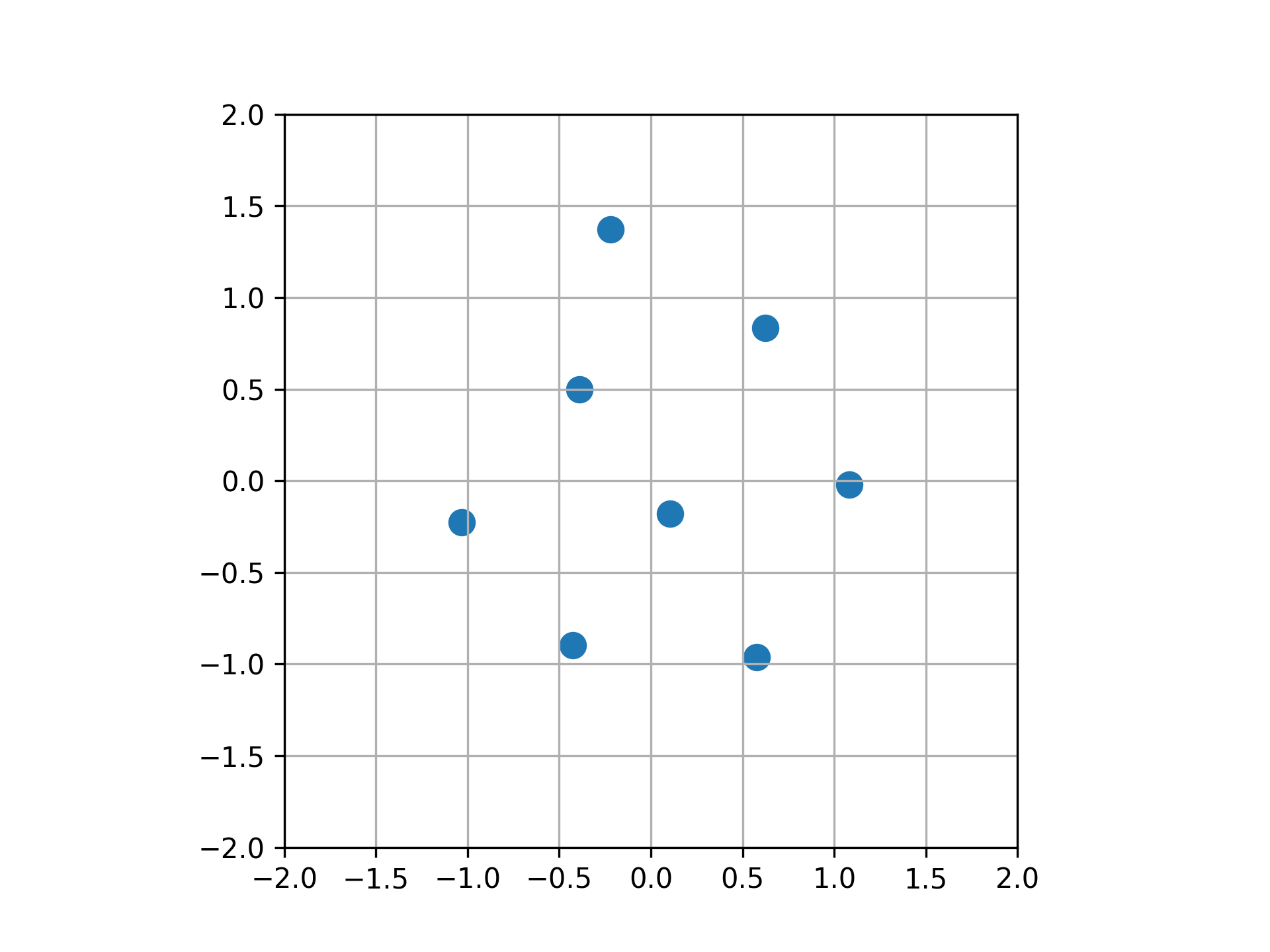}
  \subcaption{$M=8$, $\gamma_I=5$ dB}
\end{subfigure}\hfill
\begin{subfigure}[t]{0.33\textwidth}
  \centering
  \includegraphics[width=\linewidth]{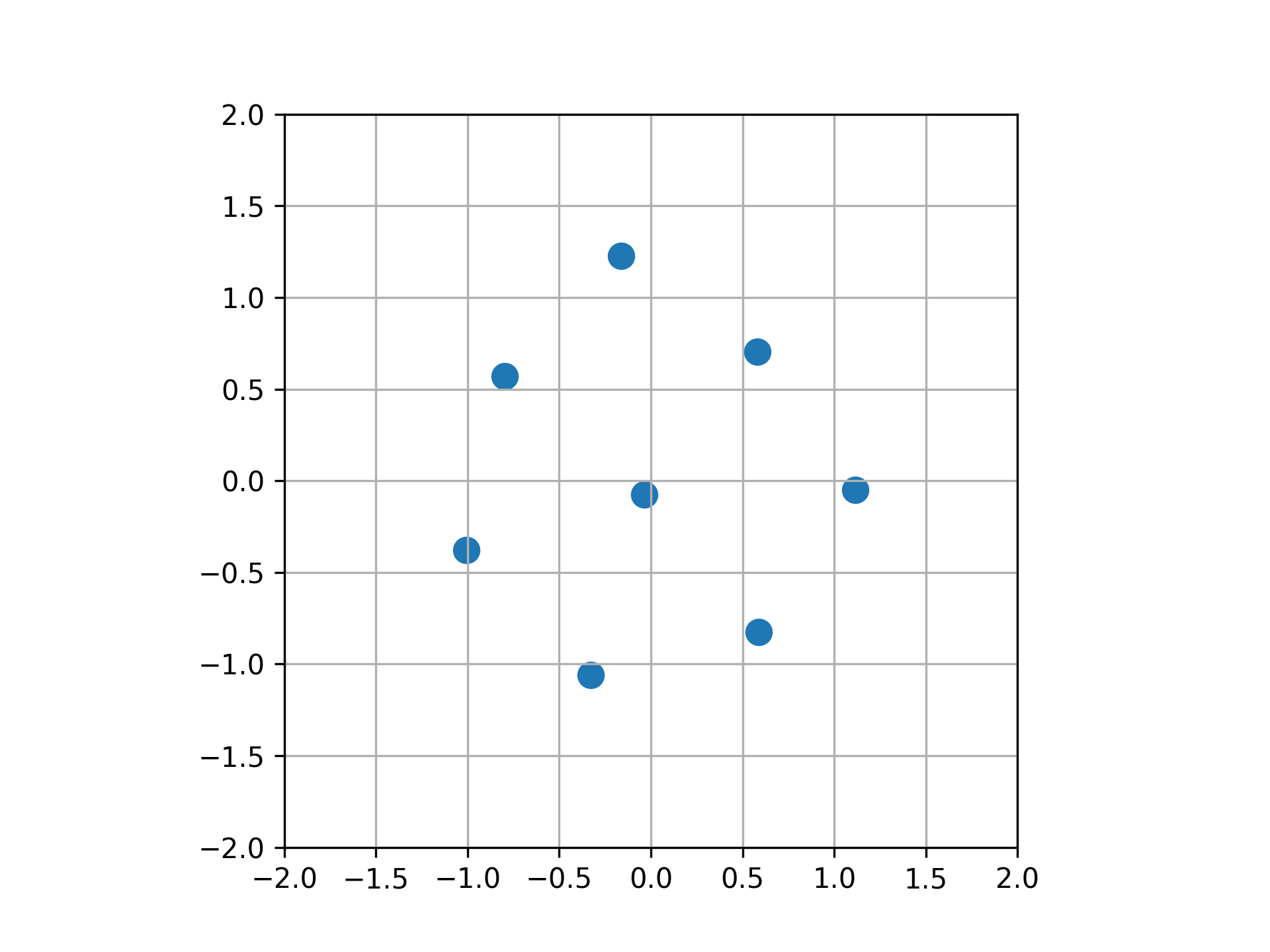}
  \subcaption{$M=8$, $\gamma_I=12$ dB}
\end{subfigure}\hfill
\begin{subfigure}[t]{0.33\textwidth}
  \centering
  \includegraphics[width=\linewidth]{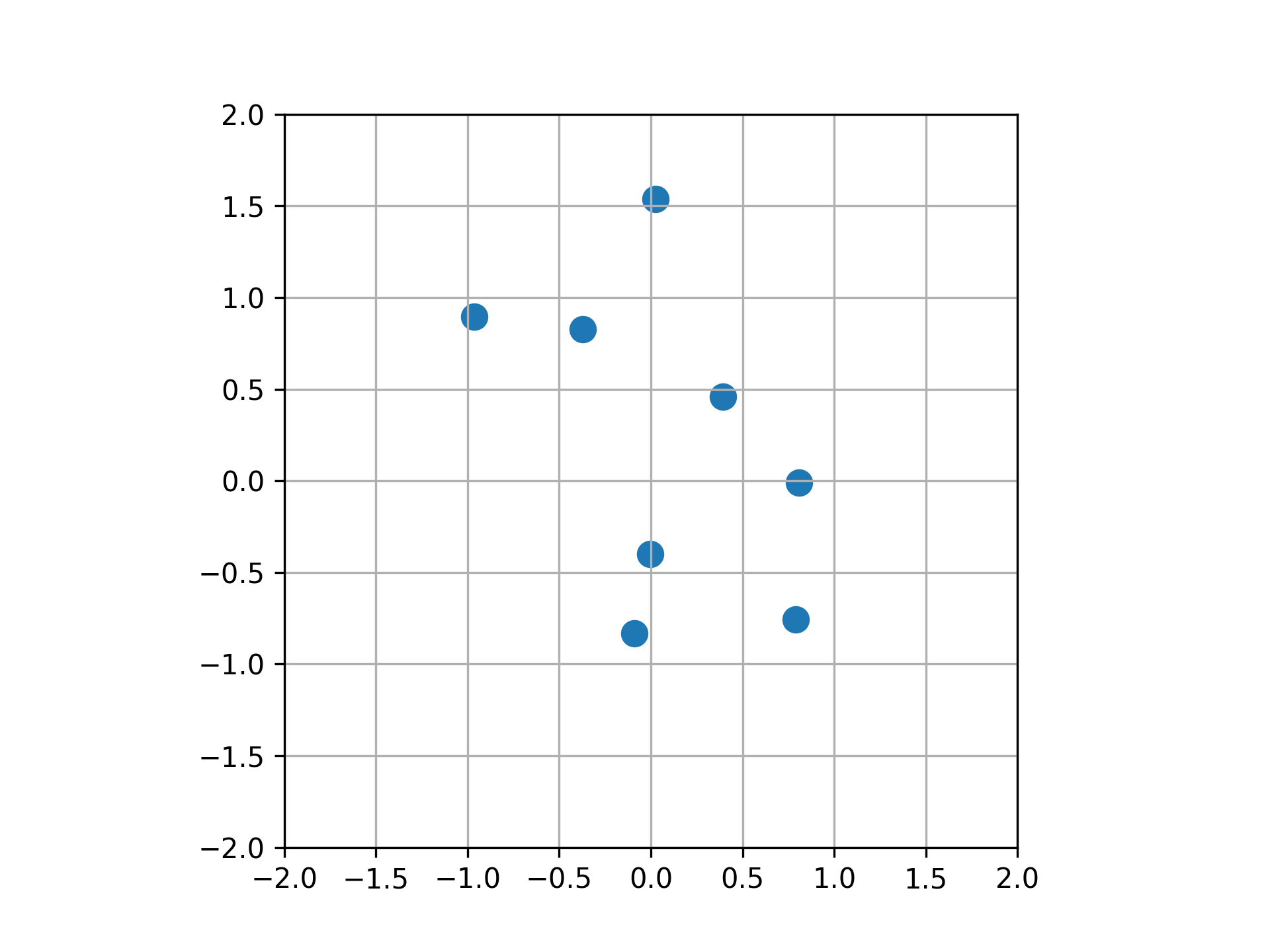}
  \subcaption{$M=8$, $\gamma_I=20$ dB}
\end{subfigure}

\vspace{4pt}

\begin{subfigure}[t]{0.33\textwidth}
  \centering
  \includegraphics[width=\linewidth]{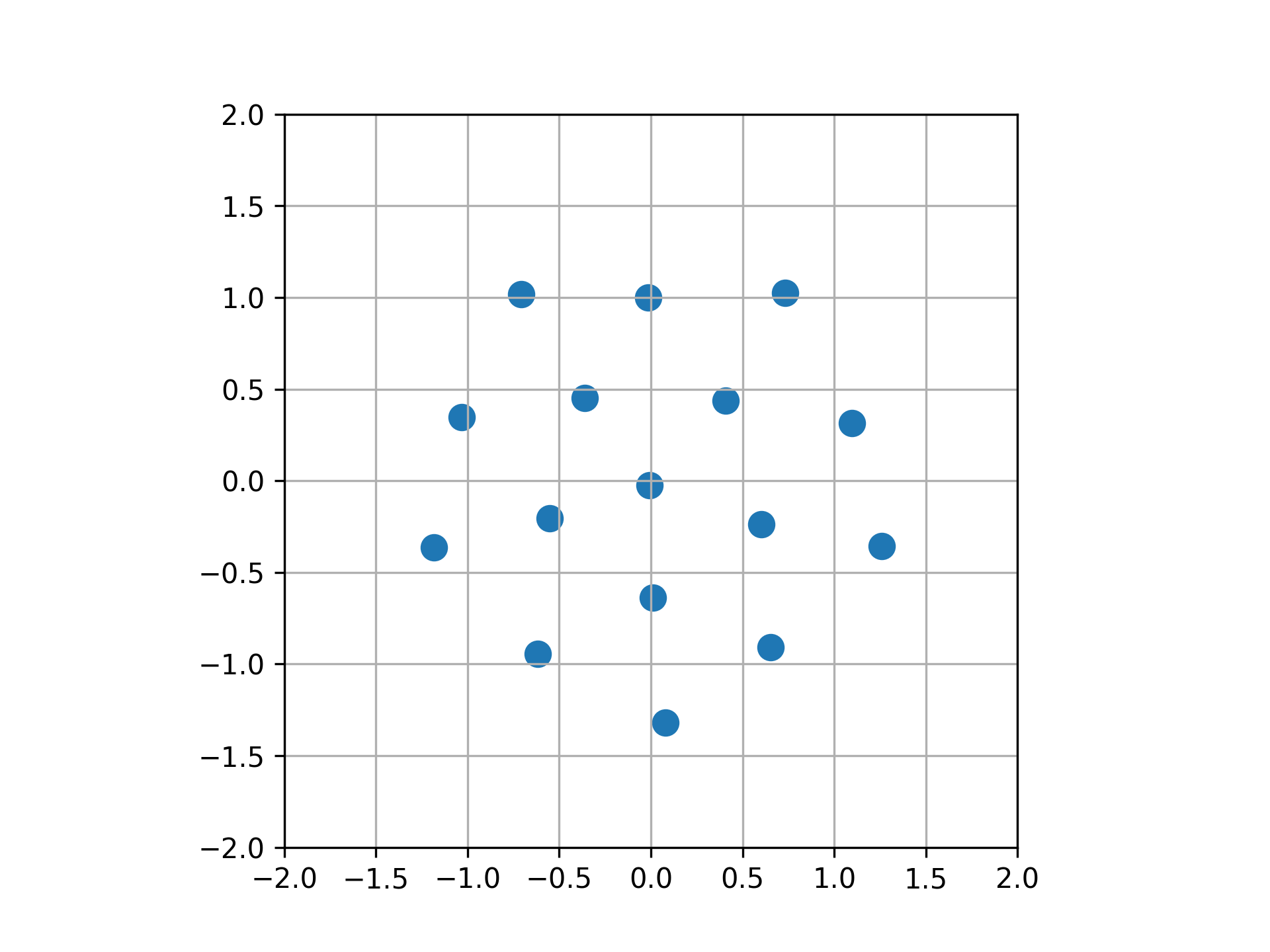}
  \subcaption{$M=16$, $\gamma_I=5$ dB}
\end{subfigure}\hfill
\begin{subfigure}[t]{0.33\textwidth}
  \centering
  \includegraphics[width=\linewidth]{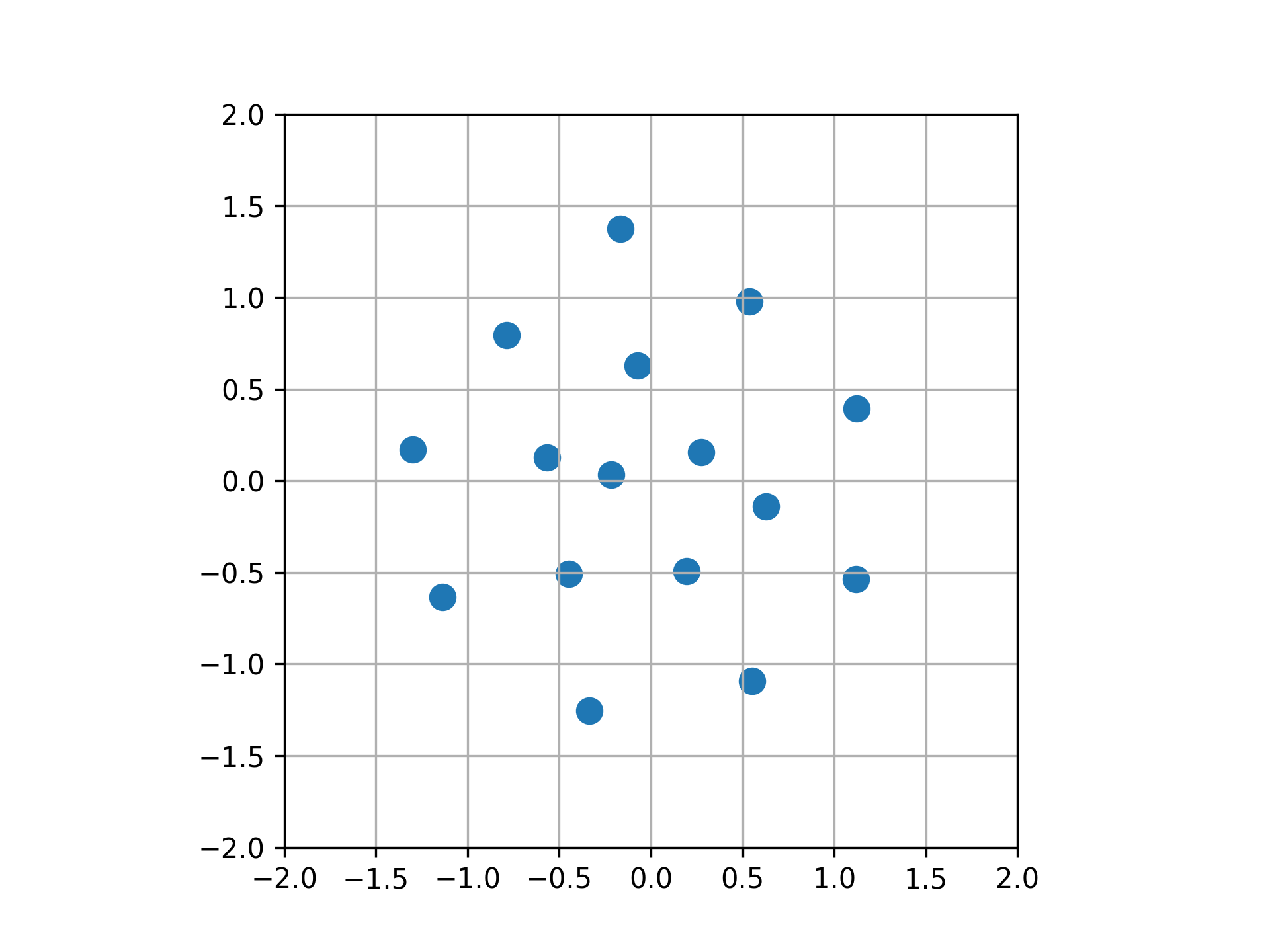}
  \subcaption{$M=16$, $\gamma_I=12$ dB}
\end{subfigure}\hfill
\begin{subfigure}[t]{0.33\textwidth}
  \centering
  \includegraphics[width=\linewidth]{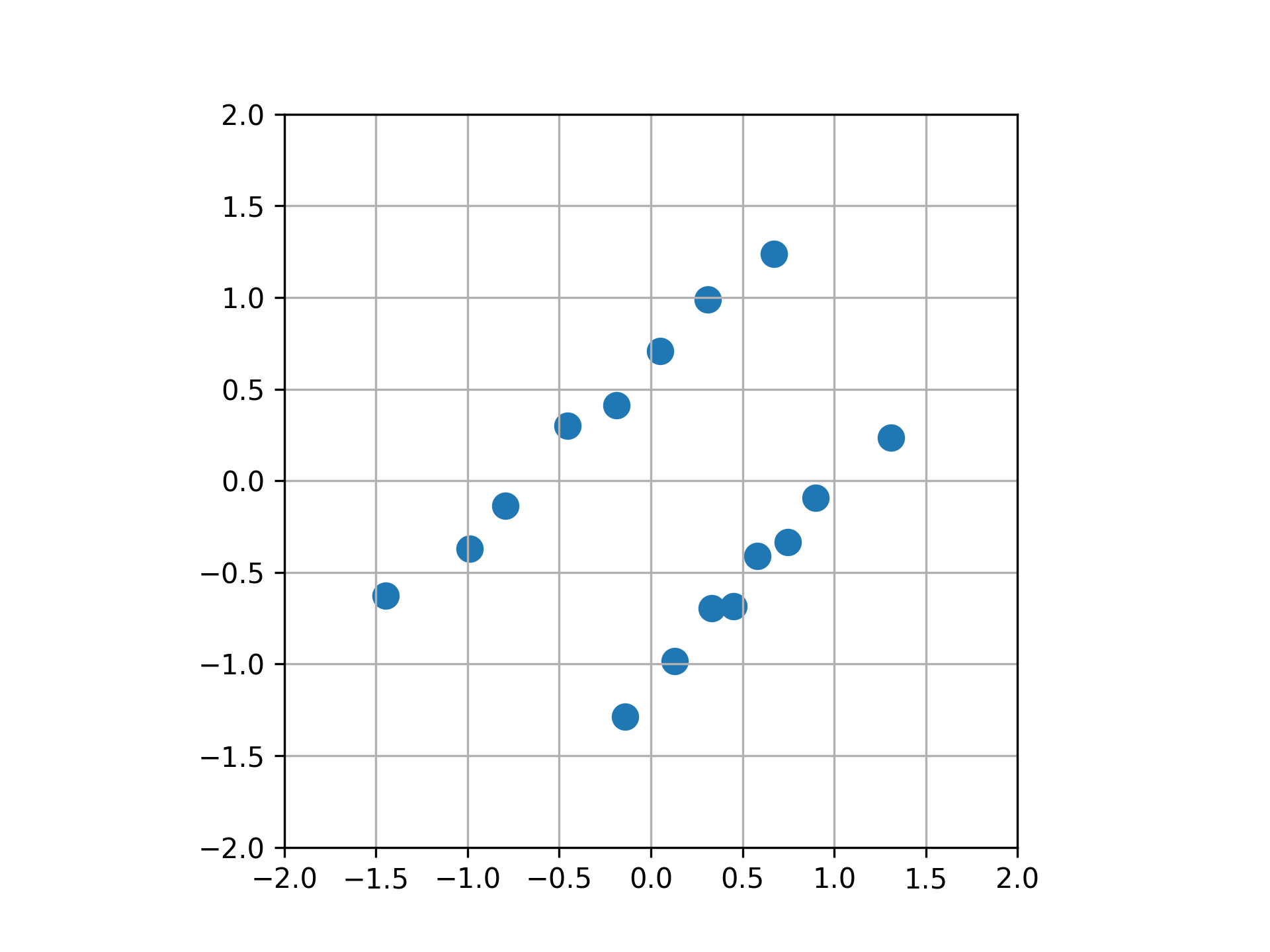}
  \subcaption{$M=16$, $\gamma_I=20$ dB}
\end{subfigure}

\caption{Optimized constellations (minimum-SEP, ML-G metric) at fixed $\gamma_S=20$ dB, and $\gamma_I \in \{5,12,20\}$ dB, $m=2$.}
\label{fig:MinSER_Constellations_Rotated}
\end{figure*}

The proposed ML-G detector induces a non-uniform decision rule whose geometry depends on both the strength and the directionality of the interfering signal. In contrast to the conventional AWGN setting, where Euclidean distance fully characterizes symbol separability, structured interference alters both the radial and angular components of the decision boundaries through the statistic $\mathcal{S}(r,\phi)$. As a result, symbol separability under the ML-G metric is governed by the joint interaction between interference statistics and constellation geometry, rather than by distance alone. This observation naturally motivates the examination of signal constellations under the proposed detection framework and leads to the formulation of constellation designs that minimize the average SEP in interference-limited conditions. Accordingly, for fixed parameters $(S,I,m)$ and constellation order $M$, the design objective is to minimize the SEP induced by the ML-G detection rule, which leads directly to the optimization problem in \eqref{eq:optSER}.

The structure of this optimization problem is dictated by the geometry of the ML-G decision regions, which are determined by the statistic $\mathcal{S}(r,\phi)$ and its harmonic series representation, and can expressed as
\begin{equation} \label{eq:optSER}
\begin{aligned}
\min_{\mathcal X} \quad & P_e(\mathcal X; S, I, m), \\
\text{s.t.} \quad & C_1: \mathcal X = {x_1,\dots,x_M},\quad x_i \in \mathbb C,\\
& C_2: \frac{1}{M}\sum_{i=1}^{M}|x_i|^{2}\le 1.
\end{aligned}
\end{equation}
In this optimization framework, the search over $\mathcal{X}$ is performed using a combination of global exploration and local refinement. Specifically, a differential-evolution algorithm~\cite{Barrueco2021} is employed to explore the feasible domain, followed by a local refinement step that enforces the average energy constraint and maintains a minimum inter-point distance to avoid degenerate solutions. Finally, since the resulting decision regions are shaped by interference-induced asymmetries and directional effects, the constellation structure is not restricted by symmetry constraints, allowing the optimization to adapt the geometry of $\mathcal{X}$ to the distortions imposed by the ML-G decision rule.


The geometric trends induced by the optimization process are reflected in Fig.~\ref{fig:MinSER_Constellations_Rotated} for $M\in{4,8,16}$ at $\gamma_S=20$ dB. At low interference levels, for example $\gamma_I=5 $ dB, the optimized constellations closely resemble hexagonal-QAM structures~\cite{thrassos_hqam}, which is consistent with the near-Euclidean behavior of the ML-G metric when angular interference effects are weak. As the interference power approaches that of the desired signal, around $\gamma_I\approx12$ dB, the optimized geometries progressively lose symmetry, reflecting the increasing influence of angular interference components and the departure from isotropic noise conditions. In the strong-interference regime, $\gamma_I\gg \gamma_S$, the optimized geometries depart substantially from conventional modulation formats, with highly irregular constellations observed for $M=4$ and $M=8$. In contrast, for $M=16$, a more structured configuration emerges, resembling two parallel 8-PAM layers oriented along nearly orthogonal axes. This geometry concentrates symbol energy along directions that are less sensitive to the dominant harmonics of $\mathcal{S}(r,\phi)$ while preserving sufficient inter-symbol separation. Overall, these results indicate a continuous transition toward interference-adaptive geometries as the INR increases, reflecting the coupled impact of envelope fading, angular interference statistics, and the non-uniform decision regions induced by the ML-G detector.

\section{Numerical Results}
\begin{figure*}[!t]
\centering
\setlength{\tabcolsep}{6pt}\renewcommand{\arraystretch}{1.2}
\begin{tabular}{cc}
\subcaptionbox{$64$-QAM for $\gamma_S=30$ dB and $m=2$.}[0.45\textwidth]{%
     \begin{tikzpicture}
    \begin{axis}[
      width=\linewidth,
      xlabel={$\gamma$ (dB)},
      ylabel={$\log_{10}(\mathrm{SEP})$},
      xmin=-20,xmax=30,
      xtick={-20,-10,0,10,20,30},
      ymin=-0.6,ymax=0,
      grid=major,
      legend pos=south west,
      legend cell align=left,
      legend style={font=\footnotesize,draw=black,fill=white},
      unbounded coords=jump,
    ]
    
    \addplot+[color=blue,line width=1.1pt,mark=*,mark size=2pt,mark repeat=6,mark phase=0,mark options={fill=blue,solid}]
    table[col sep=comma, x index=0,
          y expr={ln(max(\thisrow{SER},1e-12))/ln(10)}]
    {csv/SER_M64_SNR30dB_mu2_QAM_MLG_CAI_EuclideanDistance_smoothed_ml_g.txt};
    \addlegendentry{ML-G}
    
    \addplot+[color=orange,line width=1.1pt,dashdotted,mark=triangle*,mark size=2.5pt,mark repeat=6,mark phase=1, mark options={fill=orange,solid}]
    table[col sep=comma, x index=0,
          y expr={ln(max(\thisrow{SER},1e-12))/ln(10)}]
    {csv/SER_M64_SNR30dB_mu2_QAM_MLG_CAI_EuclideanDistance_smoothed_cai.txt};
    \addlegendentry{CAI}
    
    \addplot+[color=green!60!black,line width=1.1pt,dashed,mark=square*,mark size=2.2pt,mark repeat=6,mark phase=2, mark options={fill=green!60!black,solid}]
    table[col sep=comma, x index=0,
          y expr={ln(max(\thisrow{SER},1e-12))/ln(10)}]
    {csv/SER_M64_SNR30dB_mu2_QAM_MLG_CAI_EuclideanDistance_smoothed_eucl.txt};
    \addlegendentry{Euclidean Distance}
    
    \end{axis}
    \end{tikzpicture}
} &
    \subcaptionbox{$64$-PSK for $\gamma_S=30$ dB, $m=2$.}[0.45\textwidth]{%
      \begin{tikzpicture}
    \begin{axis}[
      width=\linewidth,
      xlabel={$\gamma$ (dB)},
      ylabel={$\log_{10}(\mathrm{SEP})$},
      xmin=-20,xmax=30,
      xtick={-20,-10,0,10,20,30},
      ymin=-0.6,ymax=0,
      grid=major,
      legend pos=south west,
      legend cell align=left,
      legend style={font=\footnotesize,draw=black,fill=white},
      unbounded coords=jump,
    ]
    
    \addplot+[color=blue,line width=1.1pt,mark=*,mark size=2pt,mark repeat=6,mark phase=0,mark options={fill=blue,solid}]
    table[col sep=comma, x index=0,
          y expr={ln(max(\thisrow{SER},1e-12))/ln(10)}]
    {csv/SER_M64_SNR30dB_mu2_PSK_MLG_CAI_EuclideanDistance_smoothed_ml_g.txt};
    \addlegendentry{ML-G}
    
    \addplot+[color=orange,line width=1.1pt,dashdotted,mark=triangle*,mark size=2.5pt,mark repeat=6,mark phase=1, mark options={fill=orange,solid}]
    table[col sep=comma, x index=0,
          y expr={ln(max(\thisrow{SER},1e-12))/ln(10)}]
    {csv/SER_M64_SNR30dB_mu2_PSK_MLG_CAI_EuclideanDistance_smoothed_cai.txt};
    \addlegendentry{CAI}
    
    \addplot+[color=green!60!black,line width=1.1pt,dashed,mark=square*,mark size=2.2pt,mark repeat=6,mark phase=2, mark options={fill=green!60!black,solid}]
    table[col sep=comma, x index=0,
          y expr={ln(max(\thisrow{SER},1e-12))/ln(10)}]
    {csv/SER_M64_SNR30dB_mu2_PSK_MLG_CAI_EuclideanDistance_smoothed_eucl.txt};
    \addlegendentry{Euclidean Distance}
    
    \end{axis}
    \end{tikzpicture}
} \\

\subcaptionbox{$64$-QAM for $\gamma_S=30$ dB, $m=5$.}[0.45\textwidth]{%
     \begin{tikzpicture}
    \begin{axis}[
      width=\linewidth,
      xlabel={$\gamma$ (dB)},
      ylabel={$\log_{10}(\mathrm{SEP})$},
      xmin=-20,xmax=30,
      xtick={-20,-10,0,10,20,30},
      ymin=-0.6,ymax=0,
      grid=major,
      legend pos=south west,
      legend cell align=left,
      legend style={font=\footnotesize,draw=black,fill=white},
      unbounded coords=jump,
    ]
    
    \addplot+[color=blue,line width=1.1pt,mark=*,mark size=2pt,mark repeat=6,mark phase=0,mark options={fill=blue,solid}]
    table[col sep=comma, x index=0,
          y expr={ln(max(\thisrow{SER},1e-12))/ln(10)}]
    {csv/SER_M64_SNR30dB_mu5_QAM_MLG_CAI_EuclideanDistance_smoothed_ml_g.txt};
    \addlegendentry{ML-G}
    
    \addplot+[color=orange,line width=1.1pt,dashdotted,mark=triangle*,mark size=2.5pt,mark repeat=6,mark phase=1, mark options={fill=orange,solid}]
    table[col sep=comma, x index=0,
          y expr={ln(max(\thisrow{SER},1e-12))/ln(10)}]
    {csv/SER_M64_SNR30dB_mu5_QAM_MLG_CAI_EuclideanDistance_smoothed_cai.txt};
    \addlegendentry{CAI}
    
    \addplot+[color=green!60!black,line width=1.1pt,dashed,mark=square*,mark size=2.2pt,mark repeat=6,mark phase=2, mark options={fill=green!60!black,solid}]
    table[col sep=comma, x index=0,
          y expr={ln(max(\thisrow{SER},1e-12))/ln(10)}]
    {csv/SER_M64_SNR30dB_mu5_QAM_MLG_CAI_EuclideanDistance_smoothed_eucl.txt};
    \addlegendentry{Euclidean Distance}
    
    \end{axis}
    \end{tikzpicture}
} &
\subcaptionbox{$64$-PSK for $\gamma_S=30$ dB, $m=5$.}[0.45\textwidth]{%
      \begin{tikzpicture}
    \begin{axis}[
      width=\linewidth,
      xlabel={$\gamma$ (dB)},
      ylabel={$\log_{10}(\mathrm{SEP})$},
      xmin=-20,xmax=30,
      xtick={-20,-10,0,10,20,30},
      ymin=-0.6,ymax=0,
      grid=major,
      legend pos=south west,
      legend cell align=left,
      legend style={font=\footnotesize,draw=black,fill=white},
      unbounded coords=jump,
    ]
    
    \addplot+[color=blue,line width=1.1pt,mark=*,mark size=2pt,mark repeat=6,mark phase=0,mark options={fill=blue,solid}]
    table[col sep=comma, x index=0,
          y expr={ln(max(\thisrow{SER},1e-12))/ln(10)}]
    {csv/SER_M64_SNR30dB_mu5_PSK_MLG_CAI_EuclideanDistance_smoothed_ml_g.txt};
    \addlegendentry{ML-G}
    
    \addplot+[color=orange,line width=1.1pt,dashdotted,mark=triangle*,mark size=2.5pt,mark repeat=6,mark phase=1, mark options={fill=orange,solid}]
    table[col sep=comma, x index=0,
          y expr={ln(max(\thisrow{SER},1e-12))/ln(10)}]
    {csv/SER_M64_SNR30dB_mu5_PSK_MLG_CAI_EuclideanDistance_smoothed_cai.txt};
    \addlegendentry{CAI}
    
    \addplot+[color=green!60!black,line width=1.1pt,dashed,mark=square*,mark size=2.2pt,mark repeat=6,mark phase=2, mark options={fill=green!60!black,solid}]
    table[col sep=comma, x index=0,
          y expr={ln(max(\thisrow{SER},1e-12))/ln(10)}]
    {csv/SER_M64_SNR30dB_mu5_PSK_MLG_CAI_EuclideanDistance_smoothed_eucl.txt};
    \addlegendentry{Euclidean Distance}
    
    \end{axis}
    \end{tikzpicture}
}
\end{tabular}

\caption{SEP versus $\gamma$ for $64$-QAM and $64$-PSK under various interference conditions.}
\label{fig:DR_SER}
\end{figure*}
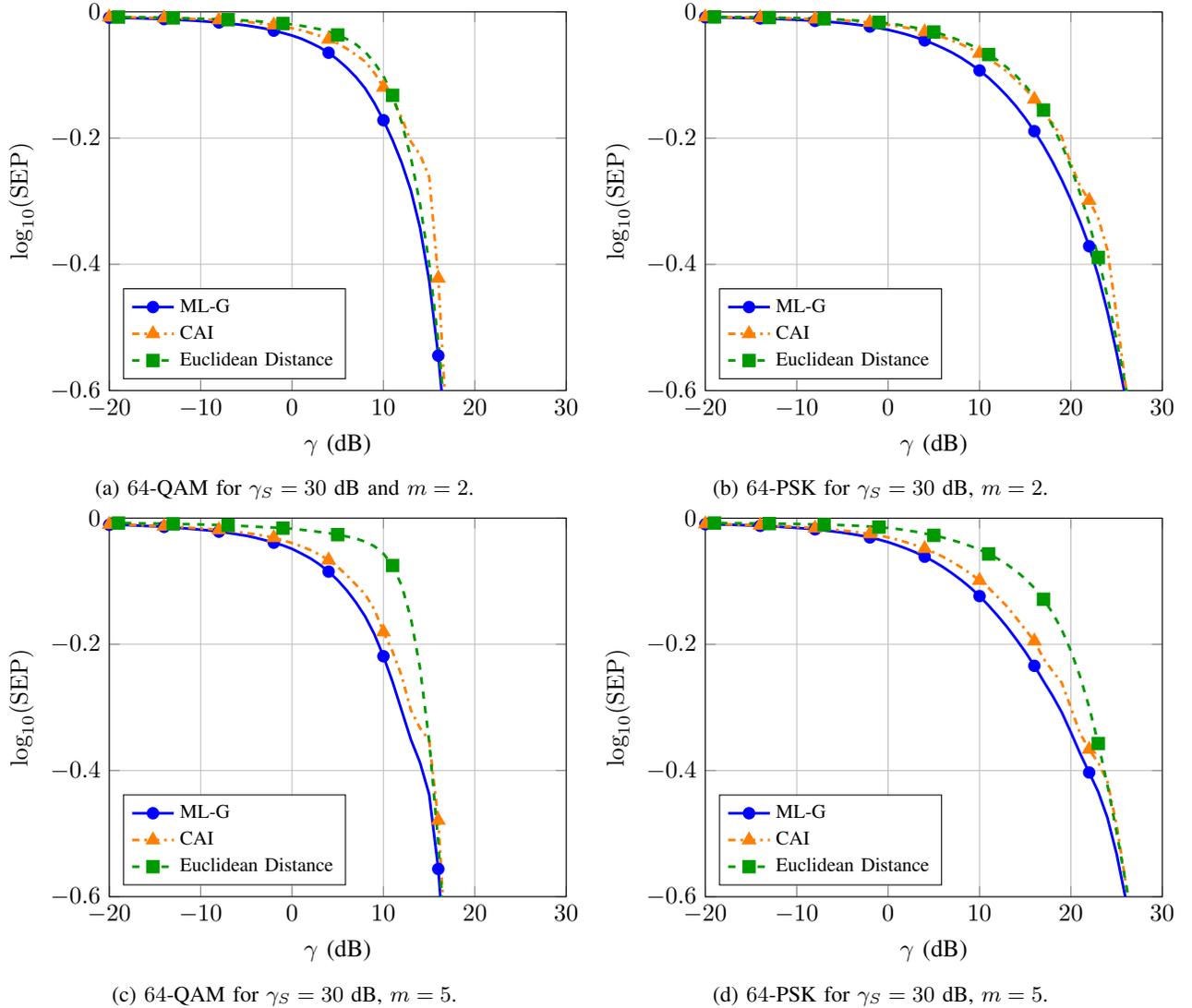

This section presents numerical evaluations of the proposed detector.  
Unless otherwise stated, performance is quantified through Monte Carlo simulation and expressed in terms of SEP.  
We examine standard constellations, including PSK, QAM, and PAM, as well as the optimized designs introduced earlier, to assess how the detector performs under varying combinations of SNR, INR, and Nakagami shape parameter $m$.  
The results highlight the relative performance of the proposed detector compared to the CAI and Euclidean baselines, clarify how interference statistics influence the geometry of the optimized constellations, and identify the interference regimes where jointly tailored detector–modulation design offers the most significant gains.

Fig.~\ref{fig:DR_SER} presents the SEP as a function of $\gamma$ for a fixed SNR of $30$ dB and shape parameters $m=2$ and $m=5$.  
When the interference dominates, i.e., $\gamma_I\gg\gamma_S$ and $\gamma$ is small, all detectors exhibit similarly high SEP, and the proposed ML-G metric offers only modest improvement over CAI and Euclidean baselines.  
As $\gamma$ increases and the interference becomes comparable to the signal, the advantage of the ML-G detector grows and reaches its maximum around $\gamma\approx 10$ dB.  
For $m=2$, the gains relative to both baselines remain similar because the interference amplitude is nearly Gaussian.  
For $m=5$, however, the improvement over the conventional Euclidean distance detector becomes more pronounced since the interference approaches a more concentrated, quasi-constant amplitude with a non-uniform phase distribution.  
In the high SNR regime, i.e., $\gamma_I\ll\gamma_S$, all detectors converge to the same SEP, which reflects the fact that the interference becomes effectively indistinguishable from Gaussian noise.

\begin{figure}[!t]
\centering
\begin{tikzpicture}
\begin{axis}[
    width=\linewidth,
    xlabel={SNR (dB)},
    ylabel={SEP},
    xmin=10, xmax=40,
    grid=major,
    legend pos=south east,
    legend cell align=left,
    legend style={font=\footnotesize},
    yticklabel style={
      /pgf/number format/fixed,
      /pgf/number format/precision=2,
      /pgf/number format/fixed zerofill
    },
    scaled y ticks=false,
]

\addplot+[color=blue,mark=*,mark size=2pt, mark repeat=3, line width=1.4pt,mark options={fill=blue,solid}] table[col sep=comma, x=x_dB, y=Ser_Diff]   {csv/MaxDiff_QAM_G_vs_Cai_smoothed_4.txt}; \addlegendentry{$M=4$}
\addplot+[color=orange, mark=+,mark size=3pt, mark repeat=3, line width=1.4pt] table[col sep=comma, x=x_dB, y=Ser_Diff]   {csv/MaxDiff_QAM_G_vs_Cai_smoothed_8.txt}; \addlegendentry{$M=8$}
\addplot+[color=green,mark=triangle*, mark size=2pt, mark repeat=3, mark options={fill=green}, line width=1.4pt] table[col sep=comma, x=x_dB, y=Ser_Diff]  {csv/MaxDiff_QAM_G_vs_Cai_smoothed_16.txt}; \addlegendentry{$M=16$}
\addplot+[color=red,mark=diamond*,mark size=2pt, mark repeat=3, line width=1.4pt] table[col sep=comma, x=x_dB, y=Ser_Diff]  {csv/MaxDiff_QAM_G_vs_Cai_smoothed_32.txt}; \addlegendentry{$M=32$}
\addplot+[color=purple,mark=triangle*, mark size=2pt, mark options={rotate=180}, mark repeat=3, line width=1.4pt]
                         table[col sep=comma, x=x_dB, y=Ser_Diff]  {csv/MaxDiff_QAM_G_vs_Cai_smoothed_64.txt}; \addlegendentry{$M=64$}
\addplot+[solid,color=brown,  mark=square*,mark size=2pt, mark repeat=3, line width=1.4pt,mark options={fill=brown,solid}] table[col sep=comma, x=x_dB, y=Ser_Diff] {csv/MaxDiff_QAM_G_vs_Cai_smoothed_128.txt}; \addlegendentry{$M=128$}

\end{axis}
\end{tikzpicture}
\caption{Maximum SEP difference between ML-G and CAI versus SNR across various QAM constellation orders.}
\label{fig:MaxDiff_QAM}
\end{figure}
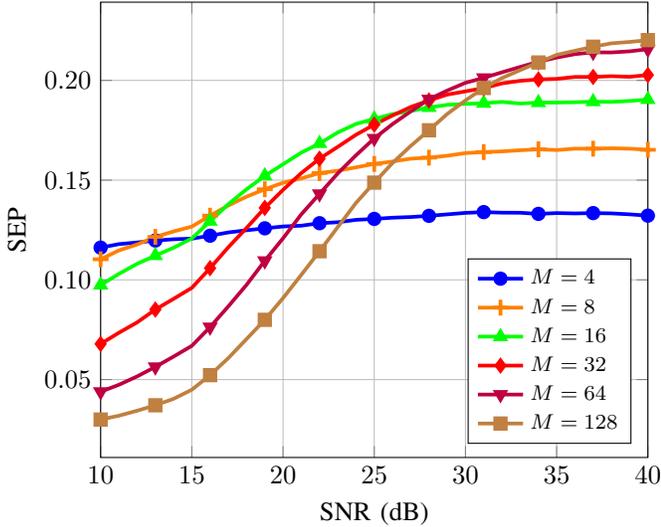

Fig. \ref{fig:MaxDiff_QAM} shows the maximum SEP improvement of the proposed ML-G detector over the CAI baseline as a function of SNR for several constellation orders. At low SNR, lower-order constellations realize the largest performance gains, reflecting their wider angular spacing and reduced sensitivity to noise. As SNR increases, the relative advantage progressively shifts toward higher-order constellations, whose finer angular resolution allows the ML-G metric to more effectively exploit the phase-dependent interference statistics. This transition highlights the interplay between constellation granularity and the degree to which the detector can capitalize on structured interference information. Overall, the results indicate that the benefit of the ML-G metric becomes increasingly pronounced when symbol separations are small and the SNR is sufficiently high for subtle phase variations to be reliably distinguished.

Fig.~\ref{fig:SER_OPT_comp} compares the SEP performance of constellations optimized for $\gamma = 0, 8,$ and $15$ dB with standard 16-PAM, 16-QAM, and 16-PSK. Among the conventional modulations, 16-PAM provides the best performance because its one-dimensional structure makes it naturally robust to phase uncertainty. As $\gamma$ increases beyond roughly $10$ dB, the constellation optimized for $\gamma = 15$ dB attains the lowest SEP, while the $\gamma = 8$ dB design offers a slight advantage in the intermediate region where interference and signal power are comparable. At very low $\gamma$, the curves nearly coincide, indicating that the channel is dominated by interference and the specific symbol geometry has little influence on detection accuracy. Overall, the results show that adapting the constellation to the interference level can provide meaningful gains, although these gains gradually diminish as the interference becomes overwhelmingly strong.

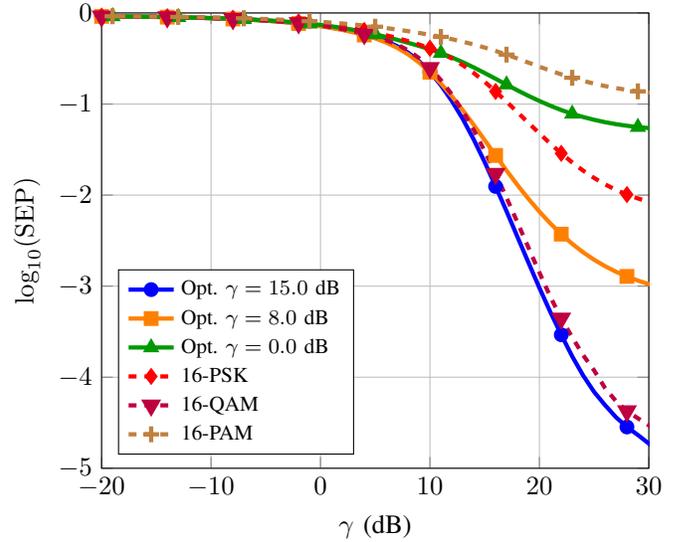
\begin{figure}[!t]
\centering
\begin{tikzpicture}
\begin{axis}[
  width=\linewidth,
  xlabel={$\gamma$ (dB)},
  ylabel={$\log_{10}(\mathrm{SEP})$},
  xmin=-20, xmax=30,
  xtick={-20,-10,0,10,20,30},
  ymin=-5, ymax=0,
  ytick={0,-1,-2,-3,-4,-5},
  grid=major,
  legend pos=south west,
  legend cell align=left,
  legend style={font=\footnotesize, draw=black, fill=white},
  unbounded coords=jump,
]

\addplot+[
  color=blue, solid, line width=1.6pt,
  mark=*, mark size=2pt, mark repeat=6, mark phase=0,
  mark options={draw=blue, fill=blue, solid}
] table[col sep=comma, x index=0,
        y expr={ln(max(\thisrow{SER},1e-12))/ln(10)}]
  {csv/SER_M16_SNR20.0dB_mu2_OptINRs_5_10_20dB_smoothed_opt5.txt};
\addlegendentry{Opt. $\gamma=15.0$ dB}

\addplot+[
  color=orange, solid, line width=1.6pt,
  mark=square*, mark size=2pt, mark repeat=6, mark phase=1,
  mark options={draw=orange, fill=orange, solid}
] table[col sep=comma, x index=0,
        y expr={ln(max(\thisrow{SER},1e-12))/ln(10)}]
  {csv/SER_M16_SNR20.0dB_mu2_OptINRs_5_10_20dB_smoothed_opt10.txt};
\addlegendentry{Opt. $\gamma=8.0$ dB}

\addplot+[
  color=green!60!black, solid, line width=1.6pt,
  mark=triangle*, mark size=2pt, mark repeat=6, mark phase=2,
  mark options={draw=green!60!black, fill=green!60!black, solid}
] table[col sep=comma, x index=0,
        y expr={ln(max(\thisrow{SER},1e-12))/ln(10)}]
  {csv/SER_M16_SNR20.0dB_mu2_OptINRs_5_10_20dB_smoothed_opt20.txt};
\addlegendentry{Opt. $\gamma = 0.0$ dB}

\addplot+[
  color=red, dashed, line width=1.6pt,
  mark=diamond*, mark size=2pt, mark repeat=6, mark phase=0,
  mark options={draw=red, fill=red, solid} 
] table[col sep=comma, x index=0,
        y expr={ln(max(\thisrow{SER},1e-12))/ln(10)}]
  {csv/SER_M16_SNR20.0dB_mu2_OptINRs_5_10_20dB_smoothed_psk.txt};
\addlegendentry{16-PSK}

\addplot+[
  color=purple, dashed, line width=1.6pt,
  mark=triangle*, mark options={rotate=180, draw=purple, fill=purple, solid},
  mark size=2.8pt, mark repeat=6, mark phase=1
] table[col sep=comma, x index=0,
        y expr={ln(max(\thisrow{SER},1e-12))/ln(10)}]
  {csv/SER_M16_SNR20.0dB_mu2_OptINRs_5_10_20dB_smoothed_qam.txt};
\addlegendentry{16-QAM}

\addplot+[
  color=brown, dashed, line width=1.6pt,
  mark=+, mark size=3.0pt, mark repeat=6, mark phase=2,
  mark options={draw=brown, fill=brown, solid}
] table[col sep=comma, x index=0,
        y expr={ln(max(\thisrow{SER},1e-12))/ln(10)}]
  {csv/SER_M16_SNR20.0dB_mu2_OptINRs_5_10_20dB_smoothed_pam.txt};
\addlegendentry{16-PAM}

\end{axis}
\end{tikzpicture}

\caption{SEP versus $\gamma$ for various constellation designs.}
\label{fig:SER_OPT_comp}
\end{figure}

\section{Conclusion}
This paper studied symbol detection and constellation design for single-carrier communication systems affected by additive interference with Nakagami-$m$ statistics, where non-circular effects alter the geometry of optimal decision regions. To address this setting, we introduced ML-G, a low-complexity detection framework that captures the impact of interference statistics while remaining suitable for practical implementation. The proposed detector induces non-uniform decision regions whose structure adapts to the interference environment and naturally reduces to conventional detection under classical fading conditions. Building on this framework, we investigated constellation design under an average-energy constraint and obtained interference-adaptive signal constellations that reflect the geometry imposed by the detection metric. Simulation results demonstrated that the proposed approach achieves consistent SEP improvements over established benchmark schemes in interference-limited regimes, while preserving robustness as operating conditions vary.


\appendices
\section{Proof of Lemma \ref{lemma:MoM}} \label{sec:appendix_MoM}
By symmetry of the PDF $g(\theta)=|{\sin 2\theta}|^{m-1}$ over the interval $[0,2\pi]$, the mean of $\Theta$ is given by
$\E[\Theta]=\pi.$
To compute the second moment, perform the change of variables $\theta=\varphi/2$ and average $\theta^2$ over one period, this yields
\begin{equation}
    \begin{aligned}\label{eq:e2steps-appendix}
    \E[\Theta^{2}]
    &=\frac{1}{16C(m)}\sum_{n=0}^{3}\int_{0}^{\pi}(\varphi^{2}+2n\pi\varphi+n^{2}\pi^{2})\sin^{m-1}d\varphi,
    \end{aligned}
\end{equation}
where the normalization constant $C(m)$ is given by
\begin{equation}
C(m)=\int_{0}^{\pi}\sin^{m-1}xdx=\sqrt{\pi}\frac{\Gamma(m/2)}{\Gamma\bigl(\tfrac{m+1}{2}\bigr)}.
\end{equation}
Define the auxiliary integrals
\begin{equation}\label{eq:Ddef-appendix}
D(m) = \int_{0}^{\pi} \varphi^{2}\sin^{m-1}\varphi d\varphi
\end{equation}
and
\begin{equation}\label{eq:Bdef-appendix}
B(m) = \int_{0}^{\pi} \varphi\sin^{m-1}\varphi d\varphi ,
\end{equation}
where the latter integral evaluates to $B(m) = \tfrac{\pi}{2} C(m)$.
Using the identities $\sum_{n=0}^{3} n = 6$ and $\sum_{n=0}^{3} n^{2} = 14$, we obtain
\begin{equation}\label{eq:e2closed-appendix}
\mathbb{E}[\Theta^{2}]
= \frac{\pi^{2}}{4} + \frac{D(m)}{4C(m)}.
\end{equation}
As can be observed, \eqref{eq:e2closed-appendix} fully characterizes the second 
moment of $\Theta$.  
For general $m>1$, the ratio $D(m)/C(m)$ does not reduce to elementary 
functions, though it can be computed numerically.

\section{Proof of Theorem \ref{prop:ML_detector}}\label{sec:appendix_ML_detector}
Consider the conditional likelihood of $Y$ given $(X,A,\Theta) = (x,a,\theta)$, which is complex Gaussian and given by
\begin{equation}\label{eq:condPDF}
  f_{Y|X,A,\Theta}(y|x,a,\theta)=\frac{1}{\pi}\exp\bigl(-|y-\sqrt{S}x-a e^{j\theta}|^{2}\bigr).
\end{equation}
Conditioning on $x$, define the residual $\rho=y-\sqrt{S}x$, with magnitude $r=|\rho|$ and phase $\phi=\arg(\rho)$.
Recenter the phase variable as $\theta=\pi+\vartheta$, which yields $\cos(\theta-\phi)=-\cos(\vartheta-\phi)$.
Approximating the phase distribution by a Gaussian with matching mean and variance, $\Theta\sim\mathcal N(\pi,\sigma_\Theta^2(m))$, the phase average evaluates to
\begin{equation}\label{eq:gauss_phase_avg}
\begin{aligned}
&\int \frac{e^{-\vartheta^{2}/(2\sigma_\Theta^{2})}}{\sqrt{2\pi}\sigma_\Theta}
e^{-2ar\cos(\vartheta-\phi)}d\vartheta\\
&= I_{0}(2ar)
+ 2\sum_{k=1}^{\infty} (-1)^{k}
e^{-\tfrac{1}{2}\sigma_\Theta^{2}k^{2}}
\cos(k\phi)I_{k}(2ar).
\end{aligned}
\end{equation}

Averaging over the interference amplitude, for $\betaN= 1+\tfrac{m}{\Omega}$, we have
\begin{equation}\label{eq:Ikm}
\begin{aligned}
I_{k,m}(r) & =
\int_{0}^{\infty} a^{2m-1} e^{-\betaN a^{2}} I_{k}(2ar)da\\
&= \frac{\Gamma\bigl(m+\tfrac{k}{2}\bigr)}{2\betaN^{m+\frac{k}{2}}k!}
r^{k}
\Fone{m+\frac{k}{2}}{k+1}{\frac{r^{2}}{\betaN}}.
\end{aligned}
\end{equation}

This result follows by substituting the series expansion $I_k(2ar)=\sum_{n\ge0}\tfrac{1}{n!\Gamma(n+k+1)}(ar)^{2n+k}$ into the defining integral, interchanging summation and integration by absolute convergence, and evaluating
\begin{equation}
\int_0^\infty a^{2(m+n)+k-1} e^{-\beta a^2}da=\tfrac{1}{2}\beta^{-(m+n)-k/2}\Gamma(m+n+k/2),
\end{equation}
where the resulting series is identified with the confluent hypergeometric function.
Collecting terms, the resulting detection rule is given by
\begin{equation}
\begin{aligned}
\hat x_{\text{ML-G}}(y)
= \argmin_{x\in\mathcal X}
\left\{|\rxx|^{2}
- \ln \mathcal{S}\left(|\rxx|,\arg(\rxx)\right)\right\},
\end{aligned}
\end{equation}
which completes the proof.

\section{Proof of Proposition \ref{proposition:existence-trunc}}
\label{sec:appendix_existence}
Consider the $k$-th term in the series defining $\mathcal{S}(r,\phi)$:
\[
T_k(r,\phi)= |w_k(\phi)| I_{k,m}(r),
\]
with
\[
|w_k(\phi)|
\le 2 e^{-\frac{1}{2}\sigma_\Theta^{2}k^{2}}.
\]
Using~\eqref{eq:Ikm},
\begin{equation}
I_{k,m}(r)
=
\frac{\Gamma\left(m+\tfrac{k}{2}\right)}{2\beta^{m+\frac{k}{2}}k!}
r^{k}
{}_1F_{1}\left(m+\frac{k}{2};k+1;\frac{r^{2}}{\beta}\right).
\end{equation}
For $r\in[0,R_{\max}]$, the hypergeometric term is $\mathcal{O}(e^{r^{2}/\beta})$
and the Gamma–factor ratio obeys
\[
\frac{\Gamma\left(m+\tfrac{k}{2}\right)}{k!}
= \mathcal{O}\left(k^{-1/2}\right).
\]
Thus,
\begin{equation}
T_k(r,\phi)
= \mathcal{O}\left(
e^{-\frac{1}{2}\sigma_\Theta^{2}k^{2}}
\frac{r^{k}}{k!}
\right),
\quad r\in[0,R_{\max}],\ \phi\in\mathbb{R}.
\end{equation}
The weight $e^{-\frac{1}{2}\sigma_\Theta^{2}k^{2}}$ ensures super-exponential decay, and therefore the series  
$\sum_{k=0}^{\infty} T_k(r,\phi)$ converges uniformly on $[0,R_{\max}]\times\mathbb{R}$.
Uniform convergence implies that for any $\varepsilon\in(0,1)$ there exists a finite integer $K$ such that
\begin{equation}
\sum_{k>K} T_k(r,\phi)\le \varepsilon,
\quad
\forall r\in[0,R_{\max}],\ \forall\phi\in\mathbb{R}.
\end{equation}
This proves the existence of a finite truncation index.

\section{Proof of Proposition \ref{proposition:min-K}}\label{sec:appendix_minK}
Starting from
\[
T_k(r,\phi)=|w_k(\phi)| I_{k,m}(r),
\]
the ratio of successive terms satisfies
\begin{equation}\label{eq:wratio-app}
\frac{|w_{k+1}(\phi)|}{|w_k(\phi)|}
\le
\exp\left[-\sigma_\Theta^{2}\left(k+\tfrac{1}{2}\right)\right].
\end{equation}
Using~\eqref{eq:Ikm}, the ratio $I_{k+1,m}(r)/I_{k,m}(r)$ becomes
\begin{equation}\label{eq:Iratio-app}
\begin{aligned}
\frac{I_{k+1,m}(r)}{I_{k,m}(r)}
=&
\frac{\Gamma\left(m+\tfrac{k+1}{2}\right)}{\Gamma\left(m+\tfrac{k}{2}\right)}
\frac{r}{(k+1)\sqrt{\beta}}\\
&\times\frac{{}_1F_{1}\left(m+\tfrac{k+1}{2};k+2;\tfrac{r^{2}}{\beta}\right)}
{{}_1F_{1}\left(m+\tfrac{k}{2};k+1;\tfrac{r^{2}}{\beta}\right)}.
\end{aligned}
\end{equation}
Bounding each factor using standard properties of $\Gamma$ and ${}_1F_1$ 
\begin{equation}
\frac{I_{k+1,m}(r)}{I_{k,m}(r)}
\le
\frac{r}{(k+1)\sqrt{\beta}}
\sqrt{m+\tfrac{k}{2}+\tfrac{1}{2}}
\exp\left(\tfrac{r^{2}}{\beta}\right).
\end{equation}
Combining with~\eqref{eq:wratio-app} we obtain
\begin{equation}\label{eq:qk-def-app}
\begin{aligned}
\frac{T_{k+1}(r,\phi)}{T_k(r,\phi)}
\le&
\exp\left[-\sigma_\Theta^{2}\left(k+\tfrac12\right)\right]
\frac{r}{(k+1)\sqrt{\beta}}
\sqrt{m+\tfrac{k}{2}+\tfrac12}
\\
&\times\exp\left(\tfrac{r^{2}}{\beta}\right)
= q_k(r).
\end{aligned}
\end{equation}
For $q_K(r)<1$ and using the geometric-tail bound,
\begin{equation}
\sum_{k>K}T_k(r,\phi)
\le
\frac{T_{K+1}(r,\phi)}{1-q_K(r)}.
\end{equation}
Bounding $T_{K+1}(r,\phi)$ with ${}_1F_1(a;b;z)\le e^{z}$ yields
\begin{equation}
T_{K+1}(r,\phi)
\le
\frac{
e^{-\frac{1}{2}\sigma_\Theta^{2}(K+1)^{2}}
r^{K+1}
\exp(r^{2}/\beta)
\Gamma\left(m+\tfrac{K+1}{2}\right)
}{
\beta^{m+\frac{K+1}{2}}(K+1)!
}.
\end{equation}
Defining $W(r,\beta,\sigma_\Theta,K)$ as in \eqref{defW},
we obtain
\begin{equation}
\sum_{k>K}T_k(r,\phi)
\le
\frac{
W(r,\beta,\sigma_\Theta,K)
\Gamma\left(m+\tfrac{K+1}{2}\right)
}{
(1-q_K(r)) \beta^{m+\frac{K+1}{2}}(K+1)!
}.
\end{equation}
Forcing this bound to be less than $\varepsilon$ over
$r\in[0,R_{\max}]$ gives the minimum truncation index $K$ guaranteed to satisfy the bound, which concludes the proof.

\bibliographystyle{IEEEtran}
\bibliography{bibliography.bib}

\end{document}